\documentclass[11pt]{article}

\usepackage{amssymb, amsmath, verbatim, amsthm,url, multirow,fullpage,mathtools, appendix}
\usepackage{longtable, rotating,makecell,array}
\usepackage[aligntableaux=top]{ytableau}

\usepackage{soul}
\usepackage[colorinlistoftodos,textsize=footnotesize]{todonotes}

\setstcolor{red}

\usepackage[all]{xy}
\xyoption{poly}
\xyoption{arc}

\usetikzlibrary{calc} 
\usetikzlibrary{decorations.pathmorphing} 
\usetikzlibrary{bending} 
\usetikzlibrary{patterns} 

\definecolor{light-gray}{gray}{0.6}

\tikzstyle{propagator}=[decorate,decoration={snake,amplitude=0.8mm}]
\tikzstyle{smallpropagator}=[decorate,decoration={snake,segment length=3mm,amplitude=0.5mm}]

\tikzstyle{linehighlight}=[blue,line width = 3pt,line cap = round, draw opacity = 0.5]

\tikzstyle{firstdash}=[dashed,line cap=round, dash pattern=on 2pt off 1pt]
\tikzstyle{seconddash}=[dashed,line cap=round, dash pattern=on 0.5pt off 1pt]
\tikzstyle{smalldash}=[dashed,line cap=round, dash pattern=on 1.5pt off 2pt]

\pgfmathsetmacro{\arrowangle}{90}
\tikzstyle{propassignment} = [->,shorten >=2pt,thick]

\newcommand{\drawWLD}[2]{

\pgfmathsetmacro{\n}{#1}
\pgfmathsetmacro{\radius}{#2}
\pgfmathsetmacro{\angle}{360/\n}
\draw (0,0) circle (\radius);
    \foreach \i in {1,2,...,\n} {
      \draw (\angle*\i:\radius) node {$\bullet$};
    }

}

\newcommand{\drawprop}[4]{
\pgfmathsetmacro{\r}{#1}
\pgfmathsetmacro{\bumpr}{#2}
\pgfmathsetmacro{\s}{#3}
\pgfmathsetmacro{\bumps}{#4}
\pgfmathsetmacro{\perturbe}{\angle/\n}
\begin{scope}
\draw[smallpropagator] (\angle*\r + \angle/2 + \bumpr*\perturbe:\radius) -- (\angle*\s + \angle/2 + \bumps*\perturbe:\radius);
\end{scope}
}

\newcommand{\drawpropbend}[5]{
\pgfmathsetmacro{\r}{#1}
\pgfmathsetmacro{\bumpr}{#2}
\pgfmathsetmacro{\s}{#3}
\pgfmathsetmacro{\bumps}{#4}
\pgfmathsetmacro{\perturbe}{\angle/\n}
\begin{scope}
\draw[smallpropagator] (\angle*\r + \angle/2 + \bumpr*\perturbe:\radius) to[bend left = #5](\angle*\s + \angle/2 + \bumps*\perturbe:\radius);
\end{scope}
}

\newcommand{\drawlabeledprop}[5]{
\pgfmathsetmacro{\r}{#1}
\pgfmathsetmacro{\bumpr}{#2}
\pgfmathsetmacro{\s}{#3}
\pgfmathsetmacro{\bumps}{#4}
\pgfmathsetmacro{\perturbe}{\angle/\n}

\begin{scope}
\draw[smallpropagator] (\angle*\r + \angle/2 + \bumpr*\perturbe:\radius) -- (\angle*\s + \angle/2 + \bumps*\perturbe:\radius) node[midway, below] {#5};
\end{scope}
}

\newcommand{\modifiedprop}[5]{
\pgfmathsetmacro{\r}{#1}
\pgfmathsetmacro{\bumpr}{#2}
\pgfmathsetmacro{\s}{#3}
\pgfmathsetmacro{\bumps}{#4}
\pgfmathsetmacro{\perturbe}{\angle/\n}

\begin{scope}
\clip (\angle*\r:\radius) -- (\angle + \angle*\r:\radius) -- (\angle*\s:\radius) -- (\angle + \angle*\s:\radius) -- (\angle*\r:\radius);
\draw[#5] (\angle*\r + \angle/2 + \bumpr*\perturbe:\radius) -- (\angle*\s + \angle/2 + \bumps*\perturbe:\radius);
\end{scope}
}

\newcommand{\boundaryprop}[4]{
\pgfmathsetmacro{\r}{#1}
\pgfmathsetmacro{\bumpr}{#2}
\pgfmathsetmacro{\s}{#3}
\pgfmathsetmacro{\perturbe}{\angle/\n}

\begin{scope}
\clip (\angle*\r:\radius) -- (\angle + \angle*\r:\radius) -- (\angle*\s - \angle:\radius) -- (\angle*\s:\radius) -- (\angle + \angle*\s:\radius) -- (\angle*\r:\radius);
\draw[#4] (\angle*\r + \angle/2 + \bumpr*\perturbe:\radius) -- (\angle*\s:\radius);
\end{scope}
	
}

\newcommand{\drawnumbers}{
  \foreach \i in {1,2,...,\n} {
  \pgfmathsetmacro{\x}{\angle*\i}
  \draw (\x:\radius*1.25) node {\footnotesize \i};
}
}

\def\centerarc[#1](#2)(#3:#4:#5)
    { \draw[#1] ($(#2)+({#5*cos(#3)},{#5*sin(#3)})$) arc (#3:#4:#5); }

\def\clipcenterarc(#1)(#2:#3:#4)
    { \clip ($(#1)+({#4*cos(#2)},{#4*sin(#2)})$) arc (#2:#3:#4); }

\newcommand{\newbetternode}[3][{label distance=-1mm]left}]{
  \node[label={#1:{\scriptsize $#3$}}] at (\zero + #2*\step:\radius) {\scriptsize $\bullet$};
}

\newcommand{\R}{\mathbb{R}}

\newcommand{\RP}{\mathbb{R}\mathbb{P}}

\newcommand{\Gr}{\mathbb{G}_{\R, \geq 0}}

\newcommand{\Grall}{\mathbb{G}_{\R}}

\newcommand{\D}{\partial}
\newcommand{\rk}{\textrm{rk }}

\def\ba #1\ea{\begin{align} #1 \end{align}}
\def\bas #1\eas{\begin{align*} #1 \end{align*}}
\def\bml #1\eml{\begin{multline} #1 \end{multline}}
\def\bmls #1\emls{\begin{multline*} #1 \end{multline*}}

\newcommand{\cP}{\mathcal{P}}
\newcommand{\cV}{\mathcal{V}}
\newcommand{\cY}{\mathcal{Y}}
\newcommand{\VP}{\cV(\cP)}
\newcommand{\YP}{\cY(\cP)}
\newcommand{\Sigmapos}{\Sigma_{\geq 0}}
\newcommand{\Lpos}{L_{\geq 0}}

\newcommand{\cA}{\mathcal{A}}
\newcommand{\cI}{\mathcal{I}}

\newcommand{\cB}{\mathcal{B}}

\newcommand{\Prop}{\textrm{Prop}}
\newcommand{\Rows}{\textrm{Row}}

\newcommand{\cW}{\mathcal{W}}

\newcommand{\cZ}{\mathcal{Z}}

\newtheorem{thm}{Theorem}[section]

\newtheorem{lem}[thm]{Lemma}
\newtheorem{cor}[thm]{Corollary}
\newtheorem{prop}[thm]{Proposition}

\theoremstyle{remark}
\newtheorem{eg}[thm]{Example}

\theoremstyle{definition}
\newtheorem{dfn}[thm]{Definition}
\newtheorem{rmk}[thm]{Remark}

\title{Cancellation of spurious poles in N=4 SYM: physical and geometric}
\author{Susama Agarwala, Cameron Marcott}

\begin{document}
\maketitle
\begin{abstract}
This paper shows that not only do the codimension one spurious poles of $N^k MHV$ tree level diagrams in N=4 SYM theory cancel in the tree level amplitude as expected, but their vanishing loci have a geometric interpretation that is tightly connected to their representation in the positive Grassmannians. In general, given a positroid variety, $\Sigma$, and a minimal matrix representation of it in terms of independent variable valued matrices, $M_\cV$, one can define a polynomial, $R(\cV)$ that is uniquely defined by the Grassmann necklace, $\cI$, of the positroid cell. The vanishing locus of $R(\cV)$ lies on the boundary of the positive variety $\overline{\Sigma} \setminus \Sigma$, but not all boundaries intersect the vanishing loci of a factor of $R(\cV)$. We use this to show that the codimension one spurious poles of  N=4 SYM, represented in twistor space, cancel in the tree level amplitude.
\end{abstract}
\section{Introduction}

The holomorphic Wilson loop representation of $N=4$ SYM theory is calculated on families of Feynman diagrams called maximally helicity violating (MHV) diagrams, next to maximal helicity violating (NMHV) diagrams, and so forth ($N^k$MHV diagrams). When represented in twistor space \cite{Adamo:2011pr,Boels:2007qn, Bullimore:2010pj}, the calculations of the associated integrals simplify dramatically. Furthermore, in a dual representation of $N^k$MHV diagrams, called Wilson loop diagrams in this paper, the diagrams correspond to subspaces of the positive Grassmanian, called positroid cells. The connection between the tree level physical interactions and the geometry of the positive Grassmannians is well studied, both in the holomorphic Wilson loop context and in the context of BCFW diagrams (or plabic graphs) where the associated geometric object is called the Amplituhedron \cite[Chapter 2]{GrassmannAmplitudebook}. 

In \cite{hodges:2013eliminating}, the author suggested that the momentum twistor representation of the integrals adopted in this paper should lead to an algebraic proof of the cancellation of spurious poles in  $N=4$ SYM theory. In this paper, not only do we show that the spurious poles cancel in the tree level amplitude,  but also that these poles are tightly connected to the positive geometry of the theory. To be specific, the spurious poles of the theory manifest as the factors of a polynomial in the denominator of the integrand associated to a Wilson loop diagram, $(\cP, [n])$ \cite{Adamo:2012xe, HeslopStewart, hodges:2013eliminating}. We indicate these polynomials as $R(\VP)$. Each Wilson loop diagram represents an $N^kMHV$ diagram, and also corresponds to a convex subspace of $\Gr(k,n)$ called a positroid cell, denoted $\Sigma(\VP)$ \cite{Wilsonloop, generalcombinatoricsI}. Any positroid cell, $\Sigma$, can also be defined by a set of Pl\"{u}cker coordinates called the Grassmann necklace, $\cI$\cite[Section 16]{Postnikov}. Recent work has shown that the irreducible factors of the product of elements of the Grassmann necklace are the frozen variables of the cluster algebra defined by $\Sigma$ \cite{galashinlam19, SS-BW}. Given a representation of $\Sigma$ in terms of variable valued matrices, $M_\cV$ one may identify a polynomial $R(\cV)$ as the radical of the polynomial formed by taking the product of the minors in the Grassmann necklace on the matrix $M_\cV$. In other words, the polynomial $R(\cV)$ corresponds to the product of the frozen variables of the cluster algebra associated to $\Sigma$ in the choice of coordinates given by $M_\cV$. Recent work has also shown that the polynomial $R(\VP)$ is a special case of this phenomenon \cite{generalcombinatoricsI}. For $\cI(\VP)$, the Grassmann necklace of $\Sigma(\VP)$, we have $R(\VP) = \textrm{rad}(\prod_{I_i \in \cI(\VP)} \Delta_{I_i}(\bf{x}))$. In other words, the spurious poles of N=4 SYM theory, as expressed in the twistor notation of Wilson loop diagrams, is the product of the frozen variable of the cluster algebra of the associated positroid cell in the set of coordinates defined by the diagram.

It is also worth noting that this cluster algebra arises from the plabic graph representation of the positoid cell. The BCFW diagrams that define the Amplituhedron have a natural representation in terms of plabic graphs \cite[Section 5]{AmplituhedronDecomposition}. While there is not yet a clear translation from the geometry of the Wilson loop diagrams to the geometry of the Amplitudehedron, this gives another tantalizing clue pointing at the similarities between the Amplituhedron and the geometry of the holomorphic Wilson loop representation. 

In this paper, we study the geometry of the spurious poles of N=4 SYM theory in terms of positroid varieties. In Theorem \ref{res:vanishonbdny}, we show that the spurious poles lie on the boundaries of the associated positroid cells. Moreover, Theorem \ref{res:polesonboundaries} shows that the poles with codimension one vanishing loci are dense in the codimension one boundary varieties. Furthermore, we show that while the vanishing loci of the square free factors of the product of the Grassmann necklace minors lies on the boundary of the corresponding positroid variety, not every codimension one boundary intersects the vanishing locus of such a factor. We give a partial characterization of which codimension one boundaries do not intersect the codimension one vanishing loci of these factors. We use these facts to show that the codimension one spurious poles of tree level Wilson loop diagrams cancel exactly in the tree level amplitude. 

Section \ref{sec:background} gives the necessary background on Wilson loop diagrams (\ref{sec:diagramdefs}) and variable valued matrices along with their associated matroids (\ref{sec:matrices and matroids}). Section \ref{sec:WLDgeom} describes the geometry of the spurious poles, with Section \ref{sec:WLDmatroid} relating the diagrams to their associated positroid varieties. Section \ref{sec:integrals} both introduces the integrals defined by the Wilson loop diagrams and the polynomial $R(\VP)$ that determines the spurious poles of N=4 SYM theory and is determined by the Grassmann necklace associated to the Wilson loop diagram. Section \ref{sec:clusteralgebras} relates the spurious poles of N=4 SYM theory to the frozen variables of the cluster algebra defined by the same positroid variety. Finally, Section \ref{sec:poles} discusses the spurious poles of the three level interactions. We show that while the variety defined by these polynomials lives on the boundary of the corresponding positroid cells, not all codimension one boundaries intersect these varieties (\ref{sec:boundarysanspoles}). Finally, in section \ref{sec:cancelation}, we also explicitly show that the codimension one spurious poles cancel exactly (Theorem \ref{res:deg1polescancel}). There are several technical lemmas that are needed in order to perform the necessary calculations to demonstrate the cancellations in  Theorem \ref{res:deg1polescancel}. The details of these calculations can be found in Appendix \ref{sec:appendix}.

\section{Diagramatical and matroidal background \label{sec:background}}

In this section, we present the necessary diagramatic background to understand Wilson loop diagrams (Section \ref{sec:diagramdefs}). Section \ref{sec:matrices and matroids} provides standard background on matroids and variable valued matrices.

\subsection{Diagramatics \label{sec:diagramdefs}}
A Wilson loop diagram $W = (\cP, [n])$ consists of a cyclically ordered set $[n] = \{1, \cdots, n\}$ and a set of propagators $\cP = \{(i,j) | i, j \in [n]\}$ as an unordered pair of integers. We depict these diagrams by drawing the set $[n]$ as vertices on a circle. The $i^{th}$ edge of the marked circle is the edge between the vertices $i$ and $i+1$. Then the propagator $p =(i,j)$ is depicted as a wavy line on the interior of the circle connecting the $i^{th}$ edge to the $j^{th}$ edge. In this manner, we say that the propagator $p$ is supported by the vertices $V_p = \{i, i+1, j, j+1\}$, as these are the vertices bounding the edges that $p$ ends on. 

It is useful to develop some nomenclature for the positioning of propagators on a given edge of a Wilson loop diagram.
  
\begin{dfn}\label{dfn:adjacentprops} Let $e$ be an edge of a Wilson Loop diagram $W = (\cP, [n])$, and  let $\{q_1, \ldots, q_s \}$ be the propagators incident on the edge $e$, ordered according to their proximity to vertex $e$. We say that $q_i$ and $q_{i+1}$ are adjacent on the edge $e$. \end{dfn}

More generally, for a subset of propagators $P \subset \cP$, we write $V_P = \cup_{p \in P} V_p$ to indicate the set of vertices supporting the propagator set $P$. For any $V \subset [n]$, the set of propagators supported by $V$ is written $\Prop(V) = \{ p \in \cP | V_p \cap V \neq \emptyset\}$.  We also give a name to the set of vertices \emph{not} supporting a set of propagators: for $P \subset \cP$, define $F(P) = V_{P^c}^c$. Next we present two examples of diagrams that we refer to throughout the paper.

\begin{eg} \label{eg:admissible}
Draw $W = (\{(3,5), (1,7)\}, [8])$ as \bas W\ =\ 
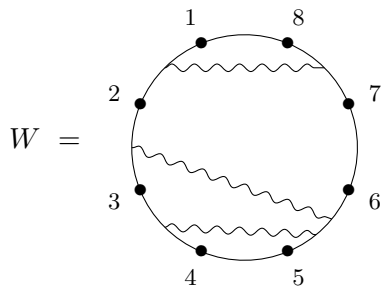
\begin{tikzpicture}[rotate=67.5,baseline=(current bounding box.east)]
	\begin{scope}
	\drawWLD{8}{1.5}
	\drawnumbers
	\drawprop{1}{0}{7}{0}
	\drawprop{3}{0}{5}{-1}
        \drawprop{2}{0}{5}{1}
		\end{scope}
	\end{tikzpicture}\eas and $W' = (\{(1,4), (3,5), (6,7), (8,1), (8,1)\}, [8])$ as \bas W'\ =\ 
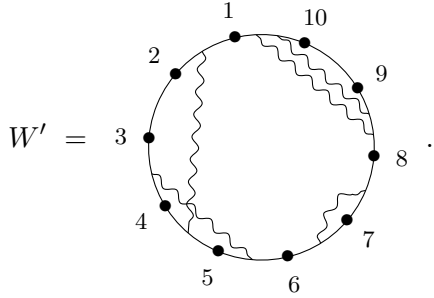
\begin{tikzpicture}[rotate=67.5,baseline=(current bounding box.east)]
	\begin{scope}
	\drawWLD{10}{1.5}
	\drawnumbers
	\drawprop{1}{0}{4}{0}
	\drawprop{3}{0}{5}{0}
        \drawpropbend{6}{0}{7}{0}{35}
	\drawprop{8}{1}{10}{-1}
 	\drawprop{8}{-2}{10}{2}
		\end{scope}
	\end{tikzpicture}\;.\eas 
Note that the pairs indicating the propagators are unordered. That is if the propagator $p$ ends on the third and fifth edges, we may write $p = (3,5)$ or $p = (5,3)$. In this paper, we use the convention that one may write $p = (i, j) = (j,i)$. That is, the two indices do not correspond to an ordered pair, or impose an orientation on $p$. Rather, they are simply the edges on which the propagators end. Furthermore, we use the convention that if two propagators have endpoints on the same edge, they are drawn so that they do not cross. 

In the diagram $W$, the propagators $p = (3,5)$ and $r = (2,5)$ both end on the the $5^{th}$  edge of the diagram. By Definition \ref{dfn:adjacentprops}, these two edges are adjacent. When we need to refer to their positioning on the $5^{th}$ edge, refer to them as $q_1 = p$ and $q_2 = r$. Consider the set of two propagators: $Q = \{p, r\}$. Then $V_Q = \{ 2, 3, 4, 5, 6\}$. Let $s = (1,7)$ be the final propagator in $W$. Then $F(s) = \{1, 7, 8\} = V_Q^c$ is the set of vertices that do not support $Q$. In particular, the vertex $2$, which is in $V_w$ and $V_r$ isn't an element of $F(s)$.
\end{eg}

The physics literature, is only interested in a certain subclass of these graphs, called admissible Wilson loop diagrams. An \emph{admissible} Wilson loop diagram, $W = (\cP, [n])$, satisfies the following conditions:
\begin{enumerate}
\item \textbf{Non-crossing} No pair of propagators $p, q \in \cP$ cross in the interior of $W$. That is, for two propagators $p = (i,j), \; q = (k, l) \in \cP$ written such that $i <j$ and $k <l$ in natural linear order on $[n]$, if $i < k$ then $l <j$. 
\item \textbf{Local Density} Any subset of propagators $ P \subset \cP$ is supported by at least 3 more vertices than the number of propagators in $P$: $|V_P| \geq |P| + 3$. 
\item \textbf{Global Density} There are at least 4 more vertices in the diagram than there are propagators. That is $n \geq |\cP| + 4$.
\end{enumerate} 

Note that the diagram $W$ above is admissible, while $W'$ is not. Furthermore, note that local density implies that one cannot have a propagators $p = (i, i+1)$ or pairs of propagators with the same endpoints: $p = q = (i, j)$.

For the remainder of this paper, we restrict our attention only to admissible Wilson loop diagrams. We denote by $\cW_{k,n}$ to be the set of all Wilson loop diagrams with $k$ propagators and $n$ vertices: $\cW_{k,n} = \{ (\cP,[n])| \textrm{ admissible }; \; |\cP| = k\}$.

\subsection{Variable valued matrices and matroids \label{sec:matrices and matroids}}

Wilson loop diagrams have a natural matrix representation with real independent variable entries. These matrices are a bridge between the combinatorics of the diagrams to matroids and to geometric subspaces of the positive Grassmannians. Before introducing this representation, we  define some notation around matrices with algebraically independent invertible variables, and recall some facts about matroids

\begin{dfn} \label{dfn:variablevaluedmatrix}
Let $\cV = \{V_1, V_2, \dots, V_k\}$ be a collection of subsets of $\{1,2,\dots,n\}$. Let $\mathbf{x}=\{x_{i,j}\}$ be a set of algebraically independent invertible variables. Define $M_{\mathcal{V}}\in M_{k,n}$ to be the $k \times n$ matrix having $x_{i,j}$ as its $i,j$ entry if $j \in V_i$ and $0$ otherwise.
\end{dfn}

One can realize a matrix $M_\cV$ at a point $\mathbf{x}=\{x_{i,j}\} \in \R^{|\mathbf{x}|}$ to get a real valued matrix. We can vary the points in $\R^{|\mathbf{x}|}$ to parameterize a family of real valued $k \times n$ matrices. Furthermore, by ignoring any such matrices of less than full rank, we may parameterize a subset of $\Grall(k,n)$.

Let $M^{\rk k}_{k,n}$ be the set of $k \times n$ matrices of full rank. Then the standard quotient map takes $M^{\rk k}_{k,n}$ to $\Grall(k,n)$ \ba \phi: M^{\rk k}_{k,n} \rightarrow \Grall(k,n)  \label{eq:maps}\ea

\begin{dfn} \label{dfn:loci}
Let $L(\cV)$ be the set of point of $\Grall(k,n)$ that can be realized by setting the entries of $M_\cV$ to real values. In other words, $L(\cV) = \phi (M_{\cV} \cap M^{\rk k}_{k,n})$.  
\end{dfn}

Finally, sometimes we need to refer to the subsets in $\cV$ that intersect a subset  $S \in \{1,2,\dots,n\}$. These sets correspond to the rows of $M_{\cV}$ have non-zero entries in the columns indicated by $S$. 

\begin{dfn}\label{dfn:columns}
For $S \in \{1,2,\dots,n\}$, write $\cV_S =\{V_i \in \cV |  V_i \cap S \neq \emptyset\}$. 
\end{dfn}

Matrices of the form $M_\cV$ define matroids\footnote{Specifically, the class of matroids defined in this way are transversal matroids \cite{Transversal}, as discussed in \cite{basisshapeloci}, but the details are beyond the scope of this discussion.}. The remainder of this section gives a brief overview of matroids, which the expert reader may skip. 

A matroid can be defined as a set, and a set of independence conditions on said set. For instance, one can define $M = (E, \mathcal{B})$, where $\mathcal{B}$, called a basis set, is a non-empty set of subsets of $E$, each of the same size, satisfying the basis exchange condition: \bas \textrm{for  all } A \textrm{ and } B \in \cB, \textrm{ if }  a \in A\setminus B \textrm{ then } \exists b \in B \setminus A \textrm{ and } (A \setminus a) \cup b \in  \cB \;.\eas Each element of $\mathcal{B}$ is a basis of $M$ and denotes a maximal independent set. The rank of the matroid, denoted $\rk(M)$ is the unique size of all the basis sets. We may also refer to the rank of a subset of $E$, $S \subset E$. We write the rank of $S$ as $\rk(S) = \max \{|B \cap S| : B \in \mathcal{B} \}$.

Equivalently, one can define $M = (E, \mathcal{F})$, where $\mathcal{F} = \{ F \subset E| \forall x \in E \setminus F, \rk(F \cup x) > \rk(F)\}$ is the set of flats of $M$. For any subset $S \subset E$, we can define the closure of the set $\textrm{cl}(S)  = \{x \in E | \rk(S) = \rk(S \cup x)\}$ as the smallest flat containing $S$. Note that if $S$ and $T$ are two flats, then $S \cap T$ is also a flat. 

Also equivalently, we may write $M = (E, \mathcal{C})$ where $\mathcal{C} = \{C \subset E | \forall S \subsetneq C, \; S \subset B \textrm{ for some } B \in \mathcal{B}\}$ is the set of circuits of $M$. The circuits are the minimal dependent sets, i.e. each proper subset of $C$ is independent. If $C$ and $D$ are both circuits, then $C \cup D$ is a cycle. 

Given a matroid $M = (E, \mathcal{B})$, and a subset $S \subset E$, the restriction $M|S = (S, \mathcal{B|S}) = \{B\cap S| B \subset \mathcal{B} \textrm{ such that } |B \cap S| = \rk(S) \}$. The contraction is defined $M/S = (E \setminus S, \mathcal{B/S}) = \{B\setminus S| B \subset \mathcal{B} \textrm{ such that } |B \cap S| \textrm{ maximal}\}$. A matroid is disconnected if it can be written as the direct sum of two matroids: $M = (E_1, \mathcal{B}_1) \oplus (E_2, \mathcal{B}_2)  = (E_1 \cup E_2 , \mathcal{B}_1 \times \mathcal{B}_2)$. Otherwise, it is connected.

A matroid is representable if it can be written as a matrix with the same independence data. In particular, since the matrices $M_\cV$ have algebraically independent non-zero entries, it is straight forward to read off the matriod associated to the collection of sets $\cV$.

\begin{eg}\label{eg:variablevaluesasmatroids}
Given a collection of subsets $\cV$ as in Definition \ref{dfn:variablevaluedmatrix}, we may define a matroid with a ground set consisting of the columns of $M_\cV$. That is, $E = \{1, \ldots , n\}$. Then, the set $S \subset E$ is a circuit in the matroid defined by $M_\cV$ if the set of rows with non-zero entries in $S$ (i.e. $\cV_S$ as defined in Definition \ref{dfn:columns}) contains exactly one fewer element than $S$: i.e. $|\cV_S| +1 < |S|$. For instance, any zero column of $M_\cV$ is a circuit. 
\end{eg}

\begin{dfn}\label{dfn:matroid}
We write $M(\cV)$ to denote the matroid defined by the variable valued matrix $M_\cV$. 
\end{dfn}

A positroid is a matroid, endowed with a cyclic ordering on the ground set, that can be realized as a matrix with all positive minors. Note that if $M$ is a positroid, as a matroid it is invariant under any permutation of the ground set. However, as a positroid, in order to  preserve the non-negativity of the minors, it is only invariant under cyclic permutations of the ground set. 

Let $<_a$ denote the $a$-th cyclic shift of the standard order on $\{1, \dots, n\}$. So, $a <_a (a+1) <_a \cdots <_a (a-1)$. These cyclic orderings on $E$ define Gale orderings on the subsets, $\{s_1 <_a s_2 <_a \cdots <_a s_k\} \leq_a \{t_1 <_a t_2 <_a \cdots <_a t_k\}$ if and only if $s_i \leq_a t_i$ for all $1 \leq i \leq k$. The collection of minimal basis sets for each cyclic shift of the Gale ordering gives the Grassmann necklace associated to a positroid. The Grassmann necklaces also define a stratification of the Grassmannians called positroid varieties. See \cite[Section 16]{Postnikov} for the original construction or \cite{positroids13} for a good exposition. Positroid varieties define subsets of $\Grall(k,n)$ that give a CW-complex when restricted to $\Gr(k,n)$ \cite[Section 3]{Postnikov}. We refer to the positroid varieties as $\Sigma$ and the positroid cells by  $\Sigmapos$.

One way to determine if a matroid is a positroid is to look at its flacets. A flacet, $F$, of a (connected) matroid $M$ is any subset of $M$ such that $M|F$ and $M/F$ are both connected. If $M$ is a positroid, then every flacet is a cyclic interval of the ground set. Furthermore, for any matroid, any flacet is a cyclic flat.

\begin{dfn}If the matroid $M(\cV)$ is a positroid, let $\Sigma(\cV)$ and $\overline{\Sigma(\cV)}$ be associated open and closed positroid varieties respectively. As such, it defines a subspace of $\Grall(k,n)$. We write $\Sigmapos(\cV) = \Sigma(\cV) \cap \Gr(k,n)$ and $\overline{\Sigmapos(\cV)} = \Sigma(\cV) \cap \Gr(k,n)$ to be the open and closed positroid cells defined by restricting to $\Gr(k,n)$. Similarly, we write $\Lpos(\cV)$ to be the restriction of $L(\cV)$ to $\Gr(k,n)$. \end{dfn} 

In this way, the word positroid has both a geometric meaning (as a subset of a Grassmannians) and a matroidal meaning (as a matroid where every flacet is a cyclic flat). These two meanings are tightly related in that every geometric positroid can be represented matroidaly as a positroid and vice versa. It should be clear from context whether the word positroid is referring to a combinatorial object or a subset of the Grassmannians. Note that while the geometrically motivated subspaces of $\Grall(k,n)$, $\Sigma(\cV)$, are in one to one correpondence with the subspaces defined by matrices, $L(\cV)$, these two subspaces are not the same.

We conclude this section with an example of how $L(\cV)$, $\Sigma(\cV)$, and $\overline{\Sigma(\cV)}$ differ.

\begin{eg} \label{eg:closuresmatch}
This example illustrates the differences between the sets $L(\cV)$, $\Sigma(\cV)$, and $\overline{\Sigma(\cV)}$. Let
\bas M_{\cV} =
\begin{bmatrix}
x_{p,1} & x_{p,2} & 0 &x_{p,4} & x_{p,5} & 0 \\
x_{q,1} & x_{q,2} & 0 & 0 & x_{q,5} & x_{q,6}
\end{bmatrix} \;. \eas  Note that this matrix corresponds to a Wilson loop diagram, but the result demonstrated in this example is much more general. Let $\Sigma(\cV)$ and $\overline{\Sigma(\cV)}$ be the associated open and closed positroid cells respectively.

The point represented by
\begin{displaymath}
\begin{bmatrix}
1 & 1 & 0 & 1 & 1 & 0 \\
1 & 1 & 0 & 0 & 1 & 1
\end{bmatrix}
\end{displaymath}
\noindent
is in $L(\cV) \setminus \Sigma(\cV)$, since the minor $ \Delta_{12}$ vanishes. 

The point represented by
\begin{displaymath}
\begin{bmatrix}
1 & 0 & 0 & 1 & 0 & 1 \\
0 & 1 & 0 & 1 & 1 & 1
\end{bmatrix}
\end{displaymath}
\noindent
is in $\Sigma(\cV) \setminus L(\cV)$ since there is no point in $L(\cV)$ where $\Delta_{45},\Delta_{56} \neq 0$ and $\Delta_{46} = 0$.

The point represented by
\begin{displaymath}
\begin{bmatrix}
0 & 0 & 0 & 0 & 1 & 0 \\
0 & 0 & 0 & 1 & 0 & 1
\end{bmatrix}
\end{displaymath}
\noindent
is in $\overline{\Sigma(\cV)} \setminus (\Sigma(\cV) \cup L(\cV))$.
\end{eg}

\section{Geometry of Wilson loop diagrams \label{sec:WLDgeom}}

We are now ready to combine the notation above to define the matrices associated to Wilson loop diagrams, and thus define the integrals associated to the diagrams. In this section, we show that the polynomials appearing in the denominator of the integrands thus defined, $R(\VP)$, are both physically and geometrically interesting. Physically, these denominators correspond to the spurious poles of $N=4$ SYM \cite{hodges:2013eliminating}. Geometrically, these polynomials are defined by the Grassmann necklace associated to the diagram \cite{generalcombinatoricsII}. Furthermore, the vanishing locus of $R(\VP)$ is contained in the boundary of the corresponding positroid variety (Theorem \ref{res:vanishonbdny}) and if the vanishing locus has codimension one in the closed positroid variety, then it is dense in the closure of a codimension one boundary Theorem \ref{res:polesonboundaries}. Finally, we show that these polynomials are closely related to the frozen variables of the cluster algebras defined by the positroid variety, $\Sigma(\VP)$. 

For physical reasons, one is only interested in the intersection of the positroid varieties with $\Gr(k,n)$. Therefore, while we prove results in this section for $\Sigma(\VP)$, for physical interpretations of the spurious poles one is only interested in the restriction to $\Sigmapos(\VP)$.

\subsection{Wilson loop diagrams and positroid varieties \label{sec:WLDmatroid}}

In order to associate a matrix to a Wilson loop diagram, $W = (\cP, [n])$, we associate a $4$ dimensional subspace of $\RP^{n+1}$ to each propagator $p \subset \cP$: \bas Y_p = \begin{cases}  x_{p,i} &  i = 0 \textrm{ or } i \in V_p \\ 0 &  \textrm{else,}\end{cases} \eas where $x_{p,i}$ are algebraically independent real valued variables \cite{Amplituhedronsquared}. Then one may associate to $W$ a subspace of $\Grall(k, n+1)$ given by the linear span of the $Y_p$. 

In other words, for each Wilson loop diagram $(\cP, [n])$, we have two sets of subsets \bas \YP = \{ V_p \cup 0 | p \in \cP\} \quad \textrm{ and } \quad \VP = \{ V_p | p \in \cP\} \eas  that define two matrices $M_{\YP}$ and $M_{\VP}$. These define subspaces of $\Grall(k, n+1)$ and $\Grall(k, n)$ respectively, that we call $L(\YP)$ and $L(\VP)$ as in Definition \ref{dfn:loci}.

\begin{eg} \label{eg:matrices}Given the Wilson loop diagram $W$ in Example \ref{eg:admissible}, with $p = (3,5)$, $r = (2,5)$, and $s = (1,7)$ one may write \bas M_{\YP} = \begin{bmatrix}  x_{p,0} & 0 & 0 &x_{p,3} & x_{p,4} &  x_{p,5} & x_{q,6} & 0 & 0 \\ x_{r,0} & 0 & x_{r,2} & x_{r,3} & 0  &x_{r,5} & x_{r,6} & 0 &0 \\ x_{s,0} & x_{s,1} & x_{s,2} & 0 &0 &0 &0 &x_{s,7} & x_{s,8} \end{bmatrix} \;.\eas

We define the matrix $M_{\VP}$ to be the one defined by ignoring the first column of $M_{\YP}$.  In the running example, \bas M_{\VP} = \begin{bmatrix}  0 & 0 &x_{p,3} & x_{p,4} &  x_{p,5} & x_{q,6} & 0 & 0 \\ 0 & x_{r,2} & x_{r,3} & 0  &x_{r,5} & x_{r,6} & 0 &0 \\  x_{s,1} & x_{s,2} & 0 &0 &0 &0 &x_{s,7} & x_{s,8} \end{bmatrix} \;.\eas
\end{eg}

\begin{rmk}
Note that, as a subspace of $\Grall(k, n)$ (resp. $\Grall(k, n+1)$), the ordering of the rows in $M_{\VP}$, resp. $M_{\YP}$ do not matter.
\end{rmk}

\begin{rmk}
It is also worth noting that in previous literature, the matrices $M_{\YP}$ and $M_{\VP}$ were denoted $C_*(W)$ and $C(W)$ respectively. However, in this paper, we wish to exploit the results of \cite{basisshapeloci} to understand the properties of the parameterized spaces and therefore change the naming conventions. 
\end{rmk}

Note that in this setup, for $W = (\cP, [n])$ and $S \subset [n]$, the set $\Prop(S)$ indicates the rows of $M_{\VP}$ that have non-zero entries in the collumns indicated by $V$. That is, in the notation of Definition \ref{dfn:columns}, $\Prop(S) = \VP_S$. 

The matroidal properties of Wilson loop diagrams derived in \cite{Wilsonloop}  can be verified by considering the independent columns of $M_{\VP}$.
\begin{enumerate} 
\item Theorem [3.9] of \cite{Wilsonloop} states that a set of vertices of a Wilson loop diagram, $(\cP, [n])$  is independent if and only if no subset supports fewer propagators than vertices in the subset. I.e. $V \subset [n]$ is independent if, for all $U \subseteq V$, $\Prop(U) \geq |U|$. In terms of rows and columns of $M_{\VP}$, this simply states that any set of columns of $M_{\VP}$ contains a dependent subset if and only if said subset is non-zero on fewer rows than columns in the subset, which is consistent with the circuit condition laid out in Example \ref{eg:variablevaluesasmatroids}
\item Corollary 3.39 of \cite{Wilsonloop} shows that all admissible Wilson loop diagrams, in particular those satisfying the non-crossing and local density conditions, correspond to positroids. This means that, as matroids, the Wilson loop diagrams can be represented by matrices with all non-negative maximal minors. That is, the subspace of $\Grall(k,n)$ parameterized by $M_{\VP}$ intersects $\Gr(k,n)$.
\end{enumerate}

\begin{eg}\label{eg:wldmatroid} For instance, in $W$  from Example \ref{eg:admissible}, the vertices $\{7,8\}$ are a circuit since they only support one propagator between them. On the other hand, the vertices $\{3,4\}$ support two propagators, and thus are independent. Consider the set of propagators $P = \{ (2, 5), (3, 5)\}$. Then $F(P) = \{ 1, 7, 8\}$. This is a flat of rank $2$. In \cite{Wilsonloop}, set of the form $F(P)$ are called propagator flats, and the author shows that the cyclic flats of the matroid associated to $W$ are propagator flats.
\end{eg}

Furthermore, \cite{basisshapeloci} gives insight into the geometry of Wilson loop diagrams in $\Gr(k,n)$. In \cite[Theorem 8.4]{basisshapeloci} the author shows that for a Wilson Loop diagram $W = (\cP, [n])$, the closure of the locus $L(\VP)$ is exactly the closure of the positroid: $\overline{L(\VP)} = \overline{\Sigma(\VP)}$. Thus, in the positive Grassmannians as well, the subspace parameterized by the matrices $M_{\VP}$ agrees with a positroid cell, up to a set of measure 0, $\overline{\Lpos(\VP)} = \overline{\Sigmapos(\VP)}$. This fact is illustrated in Example \ref{eg:closuresmatch}.

The author of \cite{basisshapeloci} also shows each $\Sigma(\VP)$ is a $3k$ dimensional space. More generally, the author shows, using slightly different notation:

\begin{thm}[Theorem 3.2, \cite{basisshapeloci}]\label{res:minimalrep}
Given a variable valued $k \times n$ matrix $M_\cV$ with $m$ non-zero entries representing a positroid  $\Sigma(\cV)$, the following are equivalent:
\begin{enumerate}
\item $\dim(\Sigma(\cV)) = m -k$ 
\item $M_\cV$ has the smallest number of non-zero variable entries of any variable valued matrices representing $\Sigma(\cV)$
\item For all $\mathcal{T} \subseteq \cV$, \bas |\bigcup_{T \in \mathcal{T}}T| \geq \max_{T \in  \mathcal{T}} (|T|) + |\mathcal{T}| -1 \;. \eas
\end{enumerate}
\end{thm}

\begin{dfn}
We say that $M_{\cV}$ is a minimal representation of $\Sigma(\cV)$ if the matroid $M(\cV)$ is a positroid and if $M_{\cV}$ satisfies condition 2 of Theorem \ref{res:minimalrep}.
\end{dfn}

Note that each matrix of the from $M_{\VP}$ is a minimal representation of $\Sigma(\VP)$. 

In \cite[section 2.3]{non-orientable}, the authors show that  the subspace of $\Grall(k,n+1)$ parameterized by $M_{\YP}$ can be viewed as a real $k$ vector bundle over the space parameterized by $M_{\VP}$, i.e.  $\cup_{W \in \cW_{k,n}}L(\YP) \rightarrow \cup_{W \in \cW_{k,n}}L(\VP)$. Restricting to any given Wilson loop diagram gives a trivial vector bundle $L(\YP) \rightarrow L(\VP)$. As shown in Theorem 2.29 of loc. cit., the space $\cup_{(\cP, [n]) \in \cW_{k,n}}L(\YP)$ need not be orientable.

\subsection{Integrals and poles \label{sec:integrals}}

The goal of this paper is to show that the spurious poles of the integrals associated to Wilson loop diagrams cancel in the calculation of the full amplitude. In this section we define the integrals and show that the poles lie on the boundary of the associated positroid cell ($\Sigmapos(\VP)$).

Recall that there are two types of singularities in this theory: the physical poles and the spurious poles. The physical poles arise when the vectors representing the particles do not span $\R^4$, as one expects in a physical theory with infrared singularities. The spurious poles are those that arise in each summand of the integrals involved in calculating the three level scattering amplitude. These should cancel in the sum. In this section, we discuss the algebraic and geometric properties  of these poles. 

\subsubsection{Integrals associated to Wilson loop diagrams}
The Wilson loop diagrams represent the tree level contributions to the scattering amplitudes in the physical theory N=4 SYM. The holomorphic Wilson loop for $n$ particles and $k$ propagators gives the contribution to the $n$ particle scattering amplitude of N=4 SYM by $k$ propagators \cite{Adamo:2011pr, Boels:2007qn, Bullimore:2010pj, hodges:2013eliminating}. The tree level contribution to this amplitude is given by a sum of integrals associated to admissible Wilson loop diagrams: \ba \cA_{k,n}^{tree} = \sum_{(\cP, [n]) \in \cW_{k,n}} I(\VP) \;.\label{eq:treelevelamplitude}\ea The scattering amplitude is a functional on the particles of the theory, represented in twistor space. In this case, the external data are represented as $n$ sections of a $k$ dimensional real vector bundle over a real twistor space. The set of external data,  $\{Z_1, \ldots, Z_n\}$, is constructed such that the $n \times (k+4)$ matrix defined with the vector $Z_i$ in the $i^{th}$ row has positive maximal ordered minors. Furthemore, we fix a gauge section, $Z_0$, which can be taken, without loss of generality, to be a $0$ section. We define the matrices \bas \cZ = \begin{bmatrix} - & Z_1& - \\ & \vdots &  \\ - & Z_n& -\end{bmatrix} \; ; \; \cZ_* = \begin{bmatrix}- & Z_0& - \\  - & Z_1& - \\ & \vdots & \\  - & Z_n& -\end{bmatrix} \; .\eas Note from above, the matrix $\cZ$ has positive maximal minors, while the matrix $\cZ_*$ may not.

Before we define the actual integral, $\cI(\VP)$, recall that if multiple propagators end on the $e^{th}$ edge, we order them according to their proximity to the vertex $e$ (Definition \ref{dfn:adjacentprops}). We are now ready to define the integrals associated to each Wilson loop diagram $W = (\cP, [n])$, with $k = |\cP|$:

\begin{dfn} \label{dfn:I(W)} \bas \cI(\VP) (\cZ_*)  = \int_{(\RP^4)^k} \frac{\prod_{p \in \cP} \prod_{v \in V_p} dx_{p, v}}{R(\VP)} \delta^{4k|4k}(M_{\YP} \cdot \cZ_*) \eas where, for $X$ a $k \times k+4$ matrix, \bas \delta^{4k|4k}(X) = \prod_{b =1}^k (X_{b, 4+b})^4\delta^4((X_{b,1},X_{b,2},X_{b,3},X_{b,4}))  \eas and $R(\VP)$ is a polynomial determined from the $W$ as follows: 
\begin{enumerate}
\item Define $R_e = x_{q_1, e+1} \big(\prod_{r = 1}^{s-1} (x_{q_r, e}x_{q_{r+1}, e+1} - x_{q_r, e+1}x_{q_{r+1}, e})\big) x_{q_s, e}$
\item $R(\VP) = \prod_{e \in [n]} R_e$
\end{enumerate} \end{dfn}

For more detail on the derivation of these integrals, see \cite{HeslopStewart, Amplituhedronsquared}. Appendix \ref{sec:appendix} gives examples and context for calculations done with these integrals. 

There are a few features to note in this definition. First, notice that the spurious poles of the theory correspond to setting the factors of $R(\VP)$ to zero. Specifically, by equations \eqref{eq:matrixvalues1} in Appendix \ref{sec:appendix}, the variables $x_{p,0}$ localize to $0$ when the vectors representing the particles do not span $\R^4$. This corresponds to the physical poles of the theory. The spurious poles occur when  $x_{p,0} \neq 0$, but the polynomial $R(\VP)$ goes to $0$. See \cite{casestudy, HeslopStewart, Amplituhedronsquared} for more calculations of this form. 

Secondly, note that in this definition we take the integral over $(\RP^4)^k$. In particular, the polynomial $R(\VP)$ is defined over $(\RP^4)^k$, and not over $L(\VP)$. Note that by a toric action, we may consider $M_{\YP}$ as a subspace of $(\RP^n)^k$, and since each row has exactly five non-zero entries, we may view each Wilson loop diagram as an isomorphism from $(\RP^4)^k$ to $M_{\YP}$. Finally, by the map in equation \eqref{eq:maps}, we may study the vanishing locus of $R(\VP)$ as a subspace of $L(\YP)$.

Thirdly, in \cite{Wilsonloop}, the authors show that the matroid $M(\VP)$ is a positroid, i.e. the parameterized space $L(\VP)$ intersects the non-negative Grassmannian. Due to the parallels with the Amplituhedron, we are interested in the non-negative geometry of the spaces defined by the Wilson loop diagrams. In \cite{basisshapeloci}, the author shows that \ba \overline{\Lpos(\VP)} = \overline{\Sigmapos(\VP)} \label{eq:density}\;.\ea  In other words, up to a space of measure $0$, the non-negative space parameterized by the Wilson loop diagrams, $\{V_p\} \subset (\RP^{n})^k$ and $M_{\VP} \subset (\R^{n \times k})$ both parameterize a particular positroid cell in the appropriate positive Grassmannian. 

Finally, returning to the main purpose of this paper, we wish to show that  the poles of these integrals (the spurious poles of the theory) cancel in the tree level amplitude. Equation \eqref{eq:treelevelamplitude} shows that $N^kMHV$ (tree level) amplitude on $n$ points is given by the sum of all the $I(\VP)$, for $(\cP, [n]) \in \cW_{k,n}$. Theorem \ref{res:polesonboundaries} shows that the spurious poles all fall on the boundaries of a geometric space defined by these diagrams. Then in Theorem \ref{res:deg1polescancel} we show that the poles that fall on codimension one boundaries cancel.

However, unlike in the Amplituhedron story, one cannot associate a geometric meaning to the sum of these integrals. In \cite{non-orientable}, the authors show that while each subspace $\Lpos(\YP)$ is orientable, the union of said spaces are not. Therefore, one cannot interpret the sum of the associated integrals as the volume of a geometric space. This result has also been shown explicitly and separately in \cite{HeslopStewart} which explicitly calculates a contradiction that occurs if one attempts this interpretation.

Here, we propose a different geometric interpretation of the sum of these integrals. The space $\Lpos(\YP)$ forms a $k$ vector bundle over $\Lpos(\VP)$ \cite{non-orientable}: $\Lpos(\YP) \rightarrow \Lpos(\VP)$. One can consider $\overline{\Lpos(\YP)}$ as the closure of a $k$-dimensional line bundle over $\overline{\Lpos(\VP)}$ \cite{non-orientable}. Then, we may consider the cancellation of spurious poles on any section of the bundle, i.e. after fixing values to the variables $x_{p, 0}$. In fact, for the spurious poles, one typically sets the values of $x_{p,0}$ to be $\pm 1$, and factors of $R(\VP)$ evaluates to $0$ \cite{casestudy, HeslopStewart, Amplituhedronsquared}. The physical poles of the tree level amplitude $A_{k,n}$ arise when the variables $x_{p, 0} = 0$.

\subsubsection{Geometry of spurious poles}

In this section, we investigate the geometry of the spurious poles of the Wilson loop diagrams, or the factors of the polynomial $R(\VP)$. We give three results showing that these factors are closely related to the geometry underlying the positroid $\Sigma(\VP)$. Namely, the polynomial is exactly the product of the square free factors of the minors identified in the Grassmann necklace of $\Sigma(\VP)$. The zero loci of the factors of $R(\VP)$ lie on the the boundary of $\Sigma(\VP)$ and the codimension one zero loci are dense in the codimension one boundaries of $\Sigma(\VP)$. 

In \cite{generalcombinatoricsII}, the authors show that the prime factors of $R(\VP)$ to are also the prime factors of the product of the minors defined by the Grassmann necklace of $\Sigma(\VP)$. We restate it here using the notation of this paper. First we introduce some notation.

\begin{dfn}
Given a variable valued matrix $M_{\VP}$, with Grassmann necklace $\cI(\VP)$, let $\Delta_{I_i}(\bf{x})$ be the polynomial in $\bf{x}$ defined by the minor of $M_{\VP}$ indicated by the columns indexed by $I_i$.
\end{dfn}

This is a polynomial in the variables making up the matrix, and is a different object than the Pl{\"u}cker coordinate indicated by $I_i$. Recall that the matrix $M_{\cV}$ is parameterized by the independent variables $\bf{x}$, with one algebraically independent variable in each non-zero entry.

\begin{thm} \cite[Proposition 5.3]{generalcombinatoricsII} \label{res:prop alg gives rad}
Write $\cI = \{I_1, \ldots I_n\}$ as the Grassmann necklace associated to the positroid cell $\Sigma(\VP)$. Then 
\begin{enumerate}
    \item Each $\Delta_{I_i}(\bf{x})$ splits into linear and quadratic factors.  All linear factors of  $\Delta_{I_i}(\bf{x})$ are single variables and all irreducible quadratic factors are $2\times 2$ determinants of single variables.
    \item Quadratic factors in $\Delta_{I_i}(\bf{x})$ arise precisely when propagators $p$ and $q$ are supported on a common edge $a$.
    \item The factor $r_e$ divides $\Delta_{I_e}(\bf{x})$.
    \item The ideal generated by $R(\VP)$ is the radical of the ideal generated by $\prod_{i=1}^{n}\Delta_{I_i}(\bf{x})$.
  \end{enumerate}
\end{thm}

We can generalize this result to define a polynomial given by a minimal parametrization of a positroid variety. 

\begin{dfn}\label{dfn:RV}
Let the set of subsets $\cV$ define a positroid variety  $\Sigma(\cV)$, with matroid $M(\cV)$. Let $\cI =  \{I_1, \ldots I_n\}$ be the associated Grassmann necklace. Then define $R(\cV)$ to be the polynomial formed by the product of the prime factors of $\prod_{i=1}^{n}\Delta_{I_i}(\bf{x})$ in the variables $\{x_{i,j}\}$ defined by the set $\cV$. That is, $R(\cV) = \textrm{rad}(\prod_{i=1}^{n}\Delta_{I_i}(\bf{x}))$. 
\end{dfn}

We note that here, $R(\cV)$ is a polynomial in the matrix entries $x_{p,i}$. The radicals taken are over the polynomial ring $\R[\bf{x}]$, defining a variety in the space of $n \times k$ matrices. To pass to the Grassmannian and obtain the variety $\overline{\Sigma(\cV)}$, one must consider the image under the map in \eqref{eq:maps}. 

In particular, there may be multiple polynomials arising from the same Grassmann necklace defining the positroid variety $\Sigma(\cV)$, as illustrated in the following example.

\begin{eg} \label{eg:differentpolys} Note that the polynomial $R(\cV)$ is dependent on the minimal parameterization of the positroid defined by the Grassmann necklace $\cI$.  For instance, in \cite{generalcombinatoricsI}, the authors give conditions for when multiple Wilson loop diagrams can correspond to the same positroid cell. For example, the two by six matrices defined by $\cV_1 = \{\{1, 2, 4, 5\}, \; \{1, 2, 3, 4\} \}$ and $\cV_2 = \{\{1, 2, 4, 5\}, \; \{2, 3, 4, 5\} \}$ both correspond to the positroid variety with Grassmann necklace $\cI = \{12, 23, 34, 45, 51, 12\}$. However \bas R(\cV_1) = x_{1,2}(x_{1,1}x_{2,2}- x_{2,1}x_{1,2})x_{2,1}x_{2, 3}x_{2,4}x_{1,4}x_{1,5}  \quad \textrm{while} \\ R(\cV_2) =  x_{1,1}x_{1,2}x_{2,2}x_{2, 3} x_{2,5}(x_{1,4}x_{2,5}- x_{2,4}x_{1,5})x_{1,4}\;.\eas Both correspond to the polynomial formed by the product of the prime factors of $\prod_{i=1}^{n}\Delta_{I_i}(\bf{x})$ under the corresponding coordinate system.

This is different from the approach elsewhere in the literature, e.g. \cite{galashinlam19, SS-BW}, where similar polynomials are defined on the Pl\"{u}cker coordinates.

\end{eg}

Next, we to show that the factors of $R(\VP)$ vanish on the boundary of the associated positroids. While this has been shown explicitly in the case of $n = 6$ and $k=2$ \cite{casestudy}, it has not been shown in general.

\begin{prop}\label{res:vanishonbdny} 
If the matroid $M(\cV)$ is a positroid, then the factors of $R(\cV)$ vanish on the boundary of the positroid variety $\Sigma(\cV)$.
\end{prop}

\begin{proof}
Let $\cI(\cV) = \{I_1, I_2, \dots, I_k\}$ be the Grassmann necklace associated $\Sigma(\cV)$. From Definition \ref{dfn:RV}, 
\begin{displaymath}
R(\cV) = \mathrm{rad}\left(\prod_{i = 1}^{k} \Delta_{I_i}(\bf{x})\right).
\end{displaymath}
From Theorem 5.15 in \cite{Juggling}, the ideal of functions defining the variety $\overline{\Sigma(\cV)}$ is generated by
\begin{displaymath}
\{\Delta_I | I_b \not \leq_b I; \; b \in [k]\},
\end{displaymath}
where $\leq_b$ indicates the $b^{th}$ cyclic shift of the Gale ordering. From \cite[Theorem 5.1]{basisshapeloci}, $L(\cV) \subset \overline{\Sigma(\cV)}$ and hence every function vanishing on $\overline{\Sigma(\cV)}$ vanishes on $L(\cV)$. Theorem 5.1 in \cite{Juggling} implies that the open positroid variety $\Sigma(\cV)$ is defined as a subset of $\overline{\Sigma(\cV)}$ by requiring certain minors to be non-vanishing. Namely, 
\begin{displaymath}
\Delta_{I_1}, \Delta_{I_2}, \dots, \Delta_{I_k} \neq 0.
\end{displaymath}
We may consider the minors $\Delta_{I_i}$ either as Pl\"{u}cker coordinates, or as polynomials in variables defining $M_{\cV}$, i.e. $\Delta_{I_i}(\bf{x})$. Using the latter interpretation, we note that the polynomial $R(\cV)$ vanishes exactly where at least one of the $\Delta_{I_i}(\bf{x})$ vanishes. So, the vanishing set of $R(\cV)$ inside $L(\cV)$ is exactly the set of points in $L(\cV)$ lying outside $\Sigma(\cV)$, or $L(\cV) \setminus \Sigma(\cV)$. Noting that
\begin{displaymath}
L(\cV) \setminus \Sigma(\cV) \subset \overline{\Sigma(\cV)} \setminus \Sigma(\cV), 
\end{displaymath} gives the desired result.
\end{proof}

Next, we use Proposition \ref{res:vanishonbdny} to show that for Wilson loop diagrams, factors of $R(\VP)$ that vanish on codimension one boundaries of $\Sigma(\VP)$ have a locus that is dense in the corresponding positroid.

\begin{thm}\label{res:polesonboundaries}
Let $W=(\mathcal{P},[n])$ be an admissible Wilson loop diagram and let $L' \subset L(\VP)$ be the vanishing locus of a single factor of $R(\VP)$ inside $L(\VP)$ and suppose that $L'$ has codimension one in $L(\VP)$. Then, the space parameterized by by the restriction is dense in a positroid variety, i.e. $\overline{L'}= \overline{\Sigma'}$, where $\Sigma'$ is a positroid cell in the boundary of $\Sigma = \Sigma(\VP)$. \end{thm}

\begin{proof}
From Theorem \ref{res:prop alg gives rad}, $R(\VP)$ is the product of individual entries and two by two minors of $M_{\VP}$. Suppose first that $L'$ is a codimension one locus obtained by setting a single variable $x_{p,i}$ to zero. Then, $L'$ is the subset of the Grassmannian consisting of row spaces of matrices of the form $M_{\VP}$ where $x_{p,i}$ is set to zero and all other entries are evaluated at real numbers.  Recall from Definition \ref{dfn:matroid} and equation \eqref{eq:maps} that all generic points in $L'$ come from the same variable valued matrix, and thus define the same matroid. Furthermore, since the basis set of the matroid associated to $L'$ is contained in basis set for $M(\cV)$, the former is a boundary of the latter. Since the positroid stratification is coarser than the matroid stratification, a generic point in $L'$ is contained in the closure of some positroid $\Sigma' \subset \overline{\Sigma}$. Proposition \ref{res:vanishonbdny} implies that $L' \subset \overline{\Sigma} \setminus \Sigma$ and so $\Sigma' \neq \Sigma$. Since $L'$ was assumed to be codimension one, the matroid represented by a generic point in $L'$ is in fact $\Sigma'$. Then, \cite[Theorem 5.1]{basisshapeloci} implies that $\overline{L'} = \overline{\Sigma'}$, as desired. 

Next, suppose that $L'$ is a codimension one boundary obtained by setting a two by two minor of $M_{\mathcal{P}}(\mathbf{x})$ to zero. As in Lemma \ref{lem:simplifyR(W)}, $L'$ may be represented by reparameterizing this two by two minor in $M_{\VP}$, then setting one of the new parameters to zero. Let $M'$ denote the variable valued matrix obtained by this change of variables. As above, all generic points in $L'$ represent the same matroid, namely the matroid whose bases are the minors $M'$ which aren't identically zero. Proposition \ref{res:vanishonbdny} implies $L' \subset \overline{\Sigma} \setminus \Sigma$. Then, since $L'$ has codimension one in $\overline{\Sigma}$, the matroid represented by a generic point in $L'$ is a positroid $\Sigma'$. The open positroid $\Sigma'$ has a parameterization via a Marsh-Reitsch matrix $R$ in the same number of variables as $M'$. Then, as in the proof of \cite[Theorem 5.1]{basisshapeloci}, $M'$ and $R$ are generically related by a change of basis matrix and thus $\overline{L'} = \overline{\Sigma'}$.
\end{proof}

\begin{eg} \label{eg:codim2}
Note that not all factors of $R(\VP)$ correspond to codimension one subspaces of $\Sigma(\VP)$. 

For instance, consider the diagram in Example \ref{eg:closuresmatch}. Setting $x_{p,4}$ to $0$ gives the matrix \bas M' =
\begin{bmatrix}
x_{p,1} & x_{p,2} & 0 &0 & x_{p,5} & 0 \\
x_{q,1} & x_{q,2} & 0 & 0 & x_{q,5} & x_{q,6}
\end{bmatrix}, \eas which can be written as $M_{\cV}$ with $\cV = \{V_1 = \{ 1, 2, 5\}, V_2 = \{1, 2, 5, 6\}\}$. Note that while $ |\bigcup_{V \in \cV}V|  = 4$,
\begin{displaymath}
\max_{V \in  \cV (|V|)} + |\cV| -1  = 4 + 2 - 1.
\end{displaymath}
Therefore the third equivalent statement of Theorem \ref{res:minimalrep} does not hold, and thus $M_{\cV}$ is not a minimal representation of a $5$ dimensional subspace of $\Grall(k,n)$. Since, by display (1) of \cite{basisshapeloci}, $5$ was an upper bound on the dimension of $L(\cV)$, setting $x_{p,4}$ must correspond to a higher codimension subspace of $\Sigma(\VP)$. In fact, direct calculation shows that this factor lies on a codimension 2 locus.

\end{eg}

We show in Section \ref{sec:boundarysanspoles} that the converse of Theorem \ref{res:polesonboundaries} does not hold. Namely, not every codimension one boundary of $\Sigma(\VP)$ contains the vanishing loci of a factor of $R(\VP)$.

\subsection{Cluster algebras, frozen variables, Grassmann necklaces \label{sec:clusteralgebras}}

We conclude this section with a brief aside, noting that the polynomial $\prod_{I_i \in \cI}I_i$ has appeared in relation to a cluster algebra associated to the positroid $\Sigma$ where $\cI= \{I_1,I_2, \dots, I_k\}$ be the associated Grassmann necklace \cite{galashinlam19, SS-BW}. Note that there is a difference in the Grassmann necklace presented in loc. cit. and this work. Namely, in  in that work, the Grassmann necklace is a written in terms of the Pl\"{u}cker coordinates coming from $\Grall(k,n)$.  In this paper, we fix a (minimal) representation, $M_{\cV}$, of $\Sigma(\cV)$. Then, instead of the Pl\"{u}cker coordinates $I_i$, we consider the minors $\Delta_{I_i}$ of $M_{\cV}$ as polynomials in $x_{p, i}$. 

In particular, in \cite{galashinlam19}, the authors consider $\cI^*$ to be the reverse Grassmann necklace associated to $\Sigma$. That is, $I^{\ast}_j$ is that maximal set in the $j^{th}$ cyclic shift of the Gale order on sets such that the Pl\"{u}cker coordinate $\Delta_{I^*_j}$ is non-vanishing on $\Sigma(\cV)$. Following Chapter 5 of \cite{Juggling} replacing Schubert varieties in the Grassmannian with reverse Schubert varieties, one sees $\Sigma$ can equivalently be defined as an intersection of reverse Schubert varieties. In \cite{SS-BW}, the authors show that each positroid variety defines a cluster algebra, and that the prime factors of $\prod_{I_i^\ast \in \cI^\ast}I_i^\ast$ are exactly the frozen variable of said cluster algebra. 

In the notation of this, paper, given a minimal representation $M_{\cV}$, we may define $\cI^\ast(\cV)$ as the reverse Grassmann necklace of the associated positroid. Define
\begin{displaymath}
R^{\ast}(\cV) = \mathrm{rad}\left(\prod_{i = 1}^{k} \Delta_{I^{\ast}_i}(\bf{x})\right)
\end{displaymath}
where $\bf{x}$ is the set of variables defining $M_{\cV}$.

In Theorem \ref{res:reversepoly}, we show that the polynomial $R^{\ast}(\cV)$ and $R(\cV)$ are equal. In this manner, we show that the polynomials $R(\cV)$ defined in this paper generate the radical ideal of the ideal generated by the product of the frozen variables of the associated cluster algebra, expressed in the coordinate system $\cV$. In particular, for $\VP$ defined by a Wilson loop diagram, the locus of the spurious poles (i.e. the vanishing locus of $R(\VP)$) is the vanishing locus of the frozen variables of the cluster algebra defined by $\Sigma(\VP)$. 

In other words, the spurious poles of the Wilson loop diagrams are intricately connected to the geometry of positroid varieties in ways that require further exploration.

\begin{thm}\label{res:reversepoly}
The two polynomials $R^{\ast}(\cV)$ and $R(\cV)$ are the same.
\end{thm}

\begin{proof}
This theorem follows from the fact that $R^{\ast}(\cV)$ and by $R(\cV)$ are both radical polynomials defining the same subvariety of the positroid variety $\Sigma(\cV)$. 

From Section 5 of \cite{Juggling}, the open positroid variety $\Sigma(\cV)$ is defined as a subset of $\overline{\Sigma(\cV)}$ by $\Delta_I \neq 0$ for all $I \in \mathcal{I}$, where $\mathcal{I}$ is the Grassmann necklace of $\Sigma(\cV)$. So, $\prod_{I \in \mathcal{I}} \Delta_I$ defines the subvariety $(\overline{\Sigma(\cV)} \setminus \Sigma(\cV)) \subset \overline{\Sigma(\cV)}$, the boundary of the open positroid inside the closed positroid. 

From Theorem 5.1 in \cite{basisshapeloci}, we know that $L(\cV) \subset \overline{\Sigma(\cV)}$. Let $V$ be the subset of $L(\cV)$ that is the image of the vanishing loci of $R(\cV)$ under the map in \eqref{eq:maps}. That is, $V = L(\cV) \cap \{\Delta_{I_i}(\mathbf{x}) = 0 | I_i \in \mathcal{I}\}$. Then,

\begin{equation} \label{eqn:rw_vanishes}
\begin{split}
V & = V \cap \overline{\Sigma(\cV)} \\
& = \big(L(\cV) \cap \{\Delta_{I_i} = 0 |  I_i \in \mathcal{I}\}\big) \cap \overline{\Sigma(\cV)} \\
& = L(\cV) \cap \big(\overline{\Sigma(\cV)} \setminus \Sigma(\cV)\big) \\
& =  L(\cV) \setminus \Sigma(\cV).
\end{split}
\end{equation}

Following Section 5 of \cite{Juggling} and replacing Schubert varieties with opposite Schubert varieties, $\Sigma(\VP)$ is similarly defined as a subset of $\overline{\Sigma(\cV)}$ by $\Delta_{I^{\ast}} \neq 0$ for all $I^{\ast} \in \mathcal{I}^{\ast}$. So, $\prod_{I^{\ast} \in \mathcal{I}^{\ast}} \Delta_{I^{\ast}}(\bf{x})$ also defines the subvariety $(\overline{\Sigma(\cV)} \setminus \Sigma(\cV)) \subset \overline{\Sigma(\cV)}$. Call $V^\ast$ the image of the vanishing loci of $R^\ast(\cV)$ under the map in \eqref{eq:maps}.

Following (\ref{eqn:rw_vanishes}), we see that $V^* = V$. That is, the images of the polynomials $R(\cV)$ and $R^{\ast}(\cV)$ vanish on the same set. By Definition \ref{dfn:RV} $R(\cV)$ is the polynomial generating the radical of the ideal generated by $\prod_\mathcal{I} \Delta_{I_i}(\bf{x})$. The polynomial $R^{\ast}(\cV)$ is radical by definition. Since  $R^{\ast}(\cV)$ and $R(\cV)$ are both radical polynomials defining the same variety, $R^{\ast}(\cV) = R(\cV)$. 
\end{proof}

\section{The poles of Wilson loop diagrams \label{sec:poles}}

Finally, we are ready to show that the codimension one singularities appearing in the integrals $I(\VP)$ cancel in the sum given in \eqref{eq:treelevelamplitude}. We do this by comparing positroid varieties with common codimension one boundaries that contain the vanishing loci of the polynomials $R(\VP)$. We first show that there are codimension one boundaries of a positroid variety $\Sigma(\cV)$ that do not contain the vanishing loci of any factor of $R(\cV)$ with codimension one. This phenomena occurs because while $L(\cV)$ is dense in $\overline{\Sigma(\cV)}$, $L(\cV)$ has empty intersection with certain boundary positroids. Since the vanishing set of $R(\VP)$ is a subset of $L(\cV)$, it will not contain these boundary positroids which do not intersect $L(\cV)$.

\subsection{Boundaries without poles \label{sec:boundarysanspoles}}

For a positroid variety $\Sigma$, every boundary corresponds to setting a (set of) elements of the Grassmann necklace to 0. Indeed one positroid variety is said to be in the boundary of another if the set of non-vanishing Pl{\"u}cker coordinates of the first are contained in the second. However, this property does not hold when one considers the minimal representations of $\Sigma$. Specifically, if $M_\cV$ is a minimal representation of $\Sigma(\cV)$, then not every codimension one boundary of $\Sigma(\cV)$ contains the vanishing locus of a factor of $R(\cV)$.

In order to reach a boundary of a positroid cell, one must send certain minors to 0 while not causing any previously vanishing minors to become positive. Sending parameters to $0$, which causes the polynomial $R(\cV)$ to vanish, is certainly one way to do this. However, the boundary of $\Sigma(\cV)$ contains other positroid cells which do not necessarily intersect the vanishing set of $R(\cV)$. In this section, we give one criterion for identifying boundaries of positroid cells that do not contain the loci of these vanishing factors. 

Here, we give a characterization of a class of variable valued matrices for which this type of boundary occurs. We consider a variable valued matrix, $M_{\cV}$, with certain conditions on its cyclic flats. Let $V, W \subset \{1, \ldots n\}$ define cyclic flats of $M_{\cV}$ such that:
\begin{enumerate}
\item neither flat has full rank ($\rk (V), \; \rk(W) <k$),
\item  the sets $V$, $W$ and $V \cup W$ are cyclic intervals in $[n]$ and
\item one cyclic flat is not contained in the other (e.g. $W \not \subset V$).
\end{enumerate}
Without loss of generality, suppose that the ranks of $V$ and $W$ are related as follows: $\rk(V \setminus W) \geq \rk(W \setminus V)$. That is, the flat $V$ has non-zero entries in at least as many rows outside of the rows $\cV_W$ as vice versa. Then, construct the variable valued matrix $M_{\cV'}$ from $M_\cV$ by zeroing out the variable entries of $(M_\cV)|_{W}$ and inserting an equal number of non-zero entries in the rows $\cV_V$  and columns $V \cup W$ such that circuits in $V$ and $W$ are preserved. In other words, for $\cV = \{V_1, \ldots V_n\}$ the set $\cV' = \{V'_1, \ldots V'_n\}$ is constructed as follows:
\begin{enumerate}
\item for every $V_i \in \cV_W$, write $V'_i =V_i \setminus W$;
\item the number of non-zero entries in $\cV'_V$ increases from that in $\cV_V$ by the number removed from $\cV_W$:
\bas
\sum_{V'_i \in \cV'_V} |V'_i| = \sum_{V_i \in \cV_V} |V_i| + \sum_{V_j \in \cV_W} |V_j \cap W|;
\eas
\item if $V_i \not \in \cV_{V\cup W}$ then $V_i = V'_i$.
\end{enumerate}
An example of this type of boundary pair is given in Example \ref{eg:strangeboundary} below.

\begin{prop}\label{res:moving variables}
Let $M_\cV$ and $M_{\cV'}$ be variable valued matrices as above, and let both have the same rank. Then $M_{\cV'}$ defines a positroid, $\Sigma(\cV')$, that lies in the boundary of the positroid $\Sigma(\cV)$.
\end{prop}

\begin{proof} 
We prove the theorem by showing that the Grassmann necklace associated to $M_{\cV'}$, $\cI(\cV')$, is different from the Grassmann necklace defined by $M_{\cV}$, $\cI(\cV)$, and that the basis set of $M_{\cV}$ contains the basis set of $M_{\cV'}$. In this way, we show that $\Sigma(\cV') \neq \Sigma(\cV)$, and that  $\Sigma(\cV')$ lies in the boundary of $\Sigma(\cV)$. 

First, we compare the Grassmann necklaces defining $\Sigma(\cV)$ and $\Sigma(\cV')$. We may read these directly off the matrices $M_{\cV}$ and $M_{\cV'}$. In particular, we wish to show that that $\cI(\cV) \neq \cI(\cV')$. To see this, first note that $\rk(V)$ (resp. $\rk(W)$) $>0$. If this were not true, then whichever flat had rank $0$ would be contained in the other flat, which contradicts our hypothesis. Since $V \cup W$ is a cyclic interval (in both $M_\cV$ and $M_{\cV'}$), let $v$ denote the first element in the cyclic interval $V\cup W$ such that $\rk(v) >0$. Let $I_v$ and $I_v'$ be the Grassmann necklace element starting at the column $v$ in $\cI(\cV)$ and $\cI(\cV')$ respectively. We have that $v \in I_v$ and $v \in I_v'$. 

Without loss of generality, assume that $V$ precedes $W$ in the cyclic interval $V \cup W$. Write $I_v = \{v = v_1, \ldots, v_r, w_1, \ldots w_s, u_1, \ldots u_t\}$, where $v_i \in V$, $w_i \in W$ and $u_i \not \in (V \cup W)$. From this, we can see that $\rk(V) = r$ and $\rk(V\cup W) = r+s$ in $M_\cV$. Furthermore, since $W$ is not a subset of $V$, we know that $s > 0$. 

Suppose, for contradiction, we have that $\cI(\cV) = \cI(\cV')$. Then we would have $I_v = I'_v$. But this would imply that $\rk(V \cup W)$ in $M_{\cV'}$ is $r+ s$. However, by construction, in $M_{\cV'}$, $\rk(V) = \rk(V\cup W)$, i.e. $r = r+s$ which contradicts the fact that $s > 0$.

To see that basis set $\cB$ of $M_\cV$ contains the basis set $\cB'$ of $M_{\cV'}$, note that any $B \in \cB$ that does not intersect $V \cup W$ is also a basis set of $\cB'$ and vice versa. Let $B' \in \cB'$ be a basis set of $M_{\cV'}$ intersecting $V\cup W$. Partition $B'$ into two sets, those elements in $V \cup W$  and those not: $B' = B'_1 \cup B'_2$ with $B'_1 \subset V \cup W$ and $B'_2 \cap (V \cup W) = \emptyset$. If $B'$ were not a basis set of $M(\cV)$ (i.e. $B' \not \in \cB$), then $B'_1$ would contain a circuit in $M_\cV$. But this would imply that a set of columns that formed a circuit in $M_{\cV}|_V$ or $M_\cV|_W$ became an independent set in $M_{\cV'}$ which violates the construction.
\end{proof}

Next we show that for pairs of matrices $M_{\cV}$ and $M_{\cV'}$ as above, the positroid cell $\Sigma(\cV')$, which lies on the boundary of the cell $\Sigma(\cV)$, does not contain the locus of a factor of $R(\cV)$. The following Theorem considers the minors $I_v$ and $I'_v$ defined in Proposition \ref{res:moving variables}. 

\begin{thm}
Given a matrix $M_{\cV}$ and $M_{\cV'}$ as above, the polynomial $R(\cV)$ does not vanish on $\Sigma(\cV')$. 
\end{thm} 

\begin{proof}
Consider the sets $I_v$ and $I'_v$ defined above, elements of the Grassmann necklaces $\cI(\cV)$ and $\cI(\cV')$ respectively. So, $I_v =  \{v = v_1, \ldots, v_r, w_1, \ldots w_s, u_1, \ldots u_t\}$, and $I_v' =  \{v = v_1, \ldots, v_r, u_1, \ldots u_{t+s}\}$ where $v_i \in V$, $w_i \in W$ and $u_i \not \in (V \cup W)$. Furthermore, let $I_w$ be the Grassmann necklace element of $\cI(\cV)$ that is minimal in the $<_w$ order. That is, $I_w =  \{ w = w_1, \ldots w_{s+m}, u_1, \ldots u_t, \ldots \}$. 

Note that since the the cycles of $V$ and $W$ are the same in both $M_{\cV}$ and $M_{\cV'}$, $I_w$ is also an element of the Grassmann necklace $\cI(\cV')$. That is, the set $\{ w_1, \ldots w_s\}$ is still the lexicographically minimal set of rank $s$ in $M_{\cV'}$ in the $<_w$ order. Note that $w$ need not be the first element of the cycle $W$ in the $<_v$ order. 

Note that the minors $\{\Delta_{I_v}, \Delta_{I_v'}, \Delta_{I_w}\}$ are non-zero in $\Sigma(\cV)$. Furthermore, setting the minor $\Delta_{I_v}$ to $0$ in $L(\cV)$, has several implications about $I_v'$ and $I_{w_1}$, depending on how this is done: 
\begin{enumerate} 
\item if $\Delta_{I_v}$ vanishes because $\{v_1, \ldots, v_r\}$ is not independent, then $\Delta_{I_v'} = 0$; 
\item if $\Delta_{I_v}$ vanishes because $\{w_1, \ldots, w_s\}$ is not independent, then $\Delta_{I_w} = 0$;
\item if $\Delta_{I_v}$ vanishes because $\{u_1, \ldots, u_t\}$ is not independent, then both $\Delta_{I_v}$ and $\Delta_{I_w} = 0$;
\item if $\Delta_{I_v}$ vanishes because $\{w_1, \ldots, w_s, u_1, \ldots, u_t\}$ is not independent but $\{w_1, \ldots, w_s\}$ is independent, then $\Delta_{I_w} = 0$;
\item if $\Delta_{I_v}$ vanishes because $\{v_1, \ldots, v_r, u_1, \ldots, u_t\}$ is not independent but $\{v_1, \ldots, v_r\}$ is independent, then $\Delta_{I_v} = 0$;
\item if $\{v_1, \ldots, v_r, w_1, \ldots, w_s\}$ is not independent, but both $\{v_1, \ldots, v_r\}$ and $\{w_1, \ldots, w_s\}$ are independent, then there is a cycle $C \subset \{v_1, \ldots, v_r, w_1, \ldots, w_s\}$ with rank greater than $0$ and bounded above by the number of shared rows between $V$ and $W$: $0 < \rk(C) \leq |\cV_W \cap \cV_V|$. Let $C$ be the maximal such cycle, $\Sigma''$ the be positroid variety defined by this independence data, and denote $c = \rk C$ in $\Sigma''$. By construction, $C$ contains elements of both $V$ and $W$. Let $I''_v = \{v_1, \ldots, v_r, w'_1, \ldots, w'_s, u_1 \ldots u_t\}$, be the element of the Grassmann necklace of $\Sigma''$ corresponding to the vertex $v$. Here $\{w'_1, \ldots, w'_s\}$ is a different set of vertices in $W$ that is not weakly less than $\{w_1, \ldots, w_s\}$ in the Gale ordering at $v$. Therefore, $I_v'' \ngeq_v I_v'$ in the Gale ordering. So, $\Delta_{I_v''}$ vanishes on $\Sigma(\cV')$ and thus $\Sigma'' \neq \Sigma(\cV')$. Since $c > 0$ and $\rk (V\cup W) > |\cV_W \cap \cV_V|$ in $M(\cV')$,  there is a basis set, $B$, of $M(\cV')$ such that $|(B \cap C)| \geq c$. Such a basis does not exist in $\Sigma''$. In other works, $\Sigma(\cV') \not \subset \Sigma''$.
\end{enumerate}

Note that if $\cV_W \cap \cV_V = \emptyset$, we cannot have $\{v_1, \ldots, v_r, w_1, \ldots, w_s\}$ but $\{v_1, \ldots, v_r\}$ and $\{w_1, \ldots, w_s\}$.

In otherwords, in $L(\cV)$, if $\Delta_{I_v} = 0$, either $\Delta_{I_v'} = 0$, or $\Delta_{I_w} = 0$, or, if neither are $0$, it defines a positroid variety that does not contain $\Sigma(\cV')$. However, in $L(\cV')$, $\Delta_{I_v}$ a uniformly zero while $\Delta_{I_{w_1}}$ and $\Delta_{I_v'}$ are not. Therefore, $L(\cV)$ does not contain the positroid cell $\Sigma(\cV')$.

By Proposition \ref{res:vanishonbdny}, we have that $R(\cV)$ vanishes on $L(\cV) \setminus \Sigma(\cV)$. Therefore, $R(\cV)$ does not vanish on $\Sigma(\cV')$. 
\end{proof}

We give an explicit example of the phenomenon characterized above. 

\begin{eg} \label{eg:strangeboundary}
The Wilson loop diagram\bas W =  
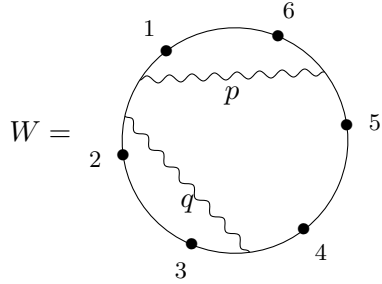
\begin{tikzpicture}[rotate=67.5,baseline=(current bounding box.east)] \begin{scope}
	\drawWLD{6}{1.5}
	\drawnumbers
	\drawlabeledprop{1}{-1}{5}{0}{$p$}
	\drawlabeledprop{1}{1}{3}{0}{$q$}
	\end{scope} \end{tikzpicture} \eas can be written as \bas \VP = \{V_p = \{1, 2, 5,6\} \; ; \; V_1 = \{1, 2, 3,4 \} \} \;.\eas with a matrix \bas M_{\VP} = \begin{bmatrix} x_{p,1} &  x_{p,2} & 0 & 0 &x_{p,5} &x_{p,6} \\x_{q,1} &  x_{q,2} & x_{q,3} &  x_{q,4}& 0 & 0 \end{bmatrix}\; . \eas The Grassmann necklace of $\Sigma(\VP)$ is $\cI(W) = \{ \{12\},\{23\}, \{35\}, \{45\}, \{51\}, \{61\} \}$. From \cite{casestudy}, we see that $\Sigma(\VP)$ shares a boundary with the positroid cells corresponding to \bas M_{\cV_1} = \begin{bmatrix} x_{1,1} &  x_{1,2} & 0 & 0 & 0 &x_{1,6} \\0 &  x_{2,2} & x_{2,3} &  x_{2,4}& x_{2,5} & x_{2,6} \end{bmatrix}  \quad \textrm{and} \quad M_{\cV_2} = \begin{bmatrix} x_{1,1} &  x_{1,2} & x_{1,3} & 0 & 0 &0 \\0 &  x_{2,2} & x_{2,3} &  x_{2,4}& x_{2,5} & x_{2,6} \end{bmatrix}  \;. \eas In particular, the common boundary is parameterized by the matrix \bas M_{\D\cV} = \begin{bmatrix} x_{1,1} &  x_{1,2} & 0 & 0 & 0 &0 \\0 &  x_{2,2} & x_{2,3} &  x_{2,4}& x_{2,5} & x_{2,6} \end{bmatrix}  \;.\eas This common boundary is 5 dimensional, parameterized for instance by setting one of the stars in each row to be 1 and allowing the other entries to be free. 

From Proposition \ref{res:vanishonbdny}, the vanishing set of $R(\VP)$ is exactly $L(\VP) \setminus \Sigma(\VP)$. We show that there is a codimension one boundary of $\Sigma(\VP)$ that does not intersect the vanishing set of $R(\VP)$. 

Let $\Sigma(\cV')$ be the positroid cell defined by the matrix $M_{\cV'}$ as defined in Theorem \ref{res:moving variables}. By construction $\Sigma(\cV')$ lies on the boundary of $\Sigma(\VP)$. We show that $\overline{\Sigma(\cV')}$ does not contain the vanishing loci of any of the factors of $R(\VP)$. The minors $\Delta_{13}$ and $\Delta_{15}$ are both non-vanishing on $\Sigma(\cV')$, while $\Delta_{35}$ is uniformly zero. However, on $L(\VP)$ the non-vanishing of the minors $\Delta_{13}$ and $\Delta_{15}$ implies the non-vanishing of the minor $\Delta_{35}$. Since $\Delta_{35}$ vanishes on $\Sigma(\cV')$, $L(\VP)$ does not intersect this boundary and hence $R(\VP)$ does not vanish on $\Sigma(\cV')$.

Note that $W = \{ 3, 4\}$, and $V = \{5, 6\}$ are two cyclic flats of $M_{\VP}$ that satisfy the conditions of the proposition. Then, up to permutations of the rows, there is only one choice for $M_{\cV'}$: \bas M_{\cV'} = \begin{bmatrix} x_{1,1} &  x_{1,2} & 0 & 0 & 0 &0 \\x_{2,1} &  x_{2,2} & x_{2,3} &  x_{2,4}& x_{2,5} & x_{2,6} \end{bmatrix} \;. \eas Note that while this is not the same matrix as $M_{\D\cV}$ these two matrices do represent the same matroid, which can be seen by comparing basis sets. In fact, the matrix $M_{\cV'}$ is a non-minimal representation of the matrix $M_{\D\cV}$.
\end{eg}

\begin{rmk}
As remarked in Example \ref{eg:strangeboundary}, the matrix $M_{\cV'}$ defined above Lemma \ref{res:moving variables} does not have the minimal number of parameters to represent the boundary positroid $\Sigma(\cV')$. In particular, we have that  \bas |\bigcup_{V \in \cV'}V| = 6 \quad \textrm{ while } \max_{V \in  \cV'} (|V|) + |\mathcal{V}| +1 = 5 + 2 +1 = 8 \;.\eas Therefore, setting $\cV' = \mathcal{T}$, we see that the third equivalence of Theorem \ref{res:minimalrep} doesn't hold, and thus $M_{\cV'}$ is not a minimal representation. In fact, there is a $GL(k)$ transformation taking $M_{\cV'}$ to $M_{\D\cV}$: \bas \begin{bmatrix}1 & 0  \\ -\frac{x_{q_1}}{x_{p_1}} & 1 \end{bmatrix}  \begin{bmatrix} x_{p,1}&  x_{p,2} & 0& 0 &0 &0 \\x_{q,1}  &  x_{q,2} & x_{q,3}&  x_{q,4}& x_{q,5} & x_{q,6} \end{bmatrix} =  \begin{bmatrix} x_{p,1}&  x_{p,2} & 0& 0 &0 &0 \\0  &  x_{q,2} -\frac{x_{q_1}x_{p_2}}{x_{p_1}} & x_{q,3}&  x_{q,4}& x_{q,5} & x_{q,6} \end{bmatrix} \;. \eas \end{rmk}

The requirement that $V$ and $W$ are cyclic intervals may seem arbitrary until one considers that all flacets of a matroid are cyclic intervals if and only if the matroid is a positroid. Thus, in some moral sense, the algorithm prescribed in Proposition \ref{res:moving variables} aims to combine two flacets into a larger flacets in order to define a new positroid cell. 

In the particular case of Wilson loop diagrams, recall from Lemma 2.28 of \cite{Wilsonloop} that all cyclic flats can be represented as propagator flats, $F(P)$, and that by Lemma 3.35 of \cite{Wilsonloop} any propagator flat that is a cyclic flat has rank equal to the number of propagators in the set defining it. Therefore, the boundaries of the sort considered in Lemma \ref{res:moving variables} occur when there are two propagators flats $F(P)$ and $F(Q)$ such that: 
\begin{enumerate} 
\item Neither $P$ nor $Q$ are the full propagators set, $\cP$. 
\item The sets $F(P)$, $F(Q)$ and $F(P) \cup F(Q)$ form cyclic intervals in $[n]$.
\item One propagator set is not contained in the other. \end{enumerate} 
It is easy to check that these conditions are met in the Wilson loop diagram $W$ in Example \ref{eg:strangeboundary}.

\subsection{Cancelation of poles on the boundary \label{sec:cancelation}}

We are now ready to prove the main result of this section: that the singularities of $I(\VP)$ that lie on codimension one boundaries of $\Sigma(\VP)$ all cancel in the tree level amplitude.  

First, recall a few facts about the polynomial $R(\VP)$. By Definition \ref{dfn:I(W)}, the primitive factors of $R(\VP)$ either have degree one or two, corresponding to one by one or two by two minors of $M_{\VP}$. In particular, if the edge $e$ supports $\{q_1, \ldots, q_s\}$ ordered as in Definition \ref{dfn:adjacentprops}, the degree one factors of $R(\VP)$ are of the form  $x_{q_1, e+1}$ or $x_{q_s, e}$. Note that there is never any factor of $R(\VP)$ that involves two non-adjacent propagators. 

\begin{thm} \label{res:deg1polescancel}
All the codimension one spulrios poles of admissible Wilson loop diagrams cancel.
\end{thm}

\begin{proof}
In the diagrams in this proof, propagators contained in a diagram are drawn with a solid line, while potential propagators are denoted with a dashed line. 

We first consider all the possible degree one factors of $R(\VP)$. Write $W = (\cP, [n])$ with $p= (i,j)  \in \cP$. There are several cases to consider:

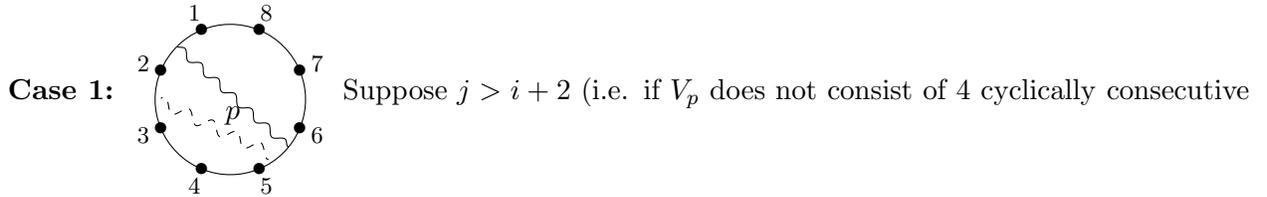
\textbf{Case 1:} $  \begin{tikzpicture}[rotate=67.5,baseline=(current bounding box.east)] \begin{scope}
	\drawWLD{8}{1}
	\drawnumbers
	\drawlabeledprop{1}{0}{5}{1}{$p$}
        \modifiedprop{2}{0}{5}{-1}{propagator, dashed}
	\end{scope} \end{tikzpicture} $ Suppose $j > i+2$ (i.e. if $V_p$ does not consist of $4$ cyclically consecutive vertices) and, without loss of generality, assume $x_{p, i}$ is a factor of $R(\VP)$. Then if $q = (i+1, j) \not \in \cP$ the we may define another diagram $W' = ((\cP \setminus p) \cup q, [n])$ that is identical to $W$ except that the propagator $p$ is replaced by $q$. Then $\lim_{x_{p, i} \rightarrow 0} I(\VP) = -\lim_{x_{q, i+2} \rightarrow 0} I(W')$, where the negative sign comes from the evaluation of $\delta^{4k|4k}$ (see Lemma \ref{lem:movingpropnegative}). For more details on the minus signs, see \cite{casestudy, HeslopStewart, Amplituhedronsquared}. By the arguments of \cite{basisshapeloci}, we see that this parametrizes a codimension one boundary of $\Sigma(\VP)$. It is easy to check that $W'$ satisfies both non-crossing (because $W$ satisfies non-crossing) and the density (because $q \not \in \cP$, and $W$ satisfies density) conditions for admissibility. Therefore $W'$ is admissible. 

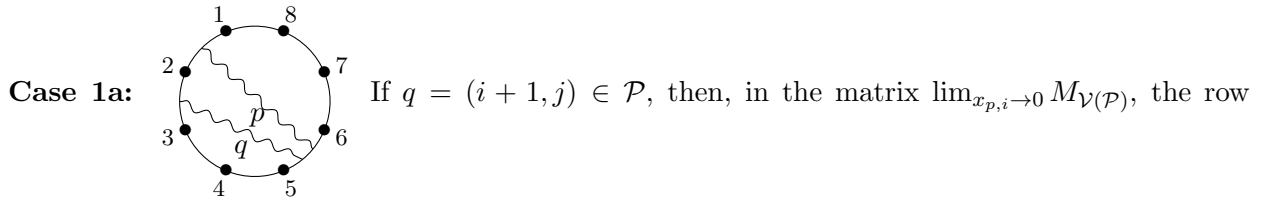
\textbf{Case 1a:} $  \begin{tikzpicture}[rotate=67.5,baseline=(current bounding box.east)] \begin{scope}
	\drawWLD{8}{1}
	\drawnumbers
	\drawlabeledprop{1}{0}{5}{1}{$p$}
        \drawlabeledprop{2}{0}{5}{-1}{$q$}
	\end{scope} \end{tikzpicture} $  If $q = (i+1, j) \in \cP$, then, in the matrix $\lim_{x_{p, i} \rightarrow 0}M_{\VP}$, the row corresponding to the propagator $p$ now has 3 non-zero entries (corresponding to the columns $\{i+1, j, j+1\}$) and the row corresponding to $q$ has $4$ non-zero entries (corresponding to the columns $\{i+1, i+2, j, j+1\}$). That is, we have $2$ rows with non-zero entries in $4$ columns. Therefore, by Theorem \ref{res:minimalrep}, we see that this locus lies in a boundary of $\Sigma(\VP)$ of codimension of at least 2. Therefore, we do not consider these poles in this argument. 

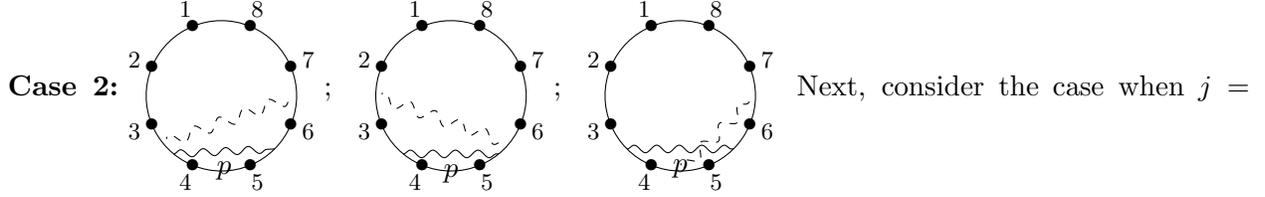
\textbf{Case 2:}$  \begin{tikzpicture}[rotate=67.5,baseline=(current bounding box.east)] \begin{scope}
	\drawWLD{8}{1}
	\drawnumbers
	\drawlabeledprop{3}{1}{5}{0}{$p$}
        \modifiedprop{3}{-1}{6}{0}{propagator, dashed}
        \end{scope} \end{tikzpicture} $;  $\begin{tikzpicture}[rotate=67.5,baseline=(current bounding box.east)] \begin{scope}
	\drawWLD{8}{1}
	\drawnumbers
	\drawlabeledprop{3}{1}{5}{-1}{$p$}
        \modifiedprop{2}{0}{5}{1}{propagator, dashed}
        \end{scope} \end{tikzpicture} $; $  \begin{tikzpicture}[rotate=67.5,baseline=(current bounding box.east)] \begin{scope}
	\drawWLD{8}{1}
	\drawnumbers
	\drawlabeledprop{3}{0}{5}{0}{$p$}
        \modifiedprop{4}{0}{6}{0}{propagator, dashed}
        \end{scope} \end{tikzpicture} $ Next, consider the case when $j = i+2$, i.e. $p = (i, i+2)$. If $x_{p, i+1}$ or $x_{p, i+2}$ is a factor of $R(\VP)$, then consider the propatator $q = (i-1, i+2)$ or $q = (i, i+3)$ respectively. If $q \not \in \cP$, then the diagram $W' = ((\cP \setminus p)\cup q, [n])$ is admissible, and the argument proceeds as Case 1. If $q \in \cP$, then the argument proceeds as in Case 1a.  If $x_{p, i}$ or $x_{p, i+3}$ is a factor of $R(\VP)$, consider $q = (i+1, i+3)$ and $q = (i-1, i+1)$ respectively. By the non-crossing condition, $p$ and $q$ cannot simultaneously exist in $W$. If $W$ does not contain another propagator of the form $(i+2, k)$ or $(i, k)$ respectively, we may define an admissible diagram $W' = ((\cP \setminus p)\cup q, [n])$ (otherwise, $q$ would cross the existing propagator $(i+2, k)$ or $(i, k)$).  In this case, $\lim_{x_{p, i} \rightarrow 0} I(\VP) = -\lim_{x_{q, i+4} \rightarrow 0} I(W')$ where the negative sign again comes from Lemma \ref{lem:movingpropnegative}.

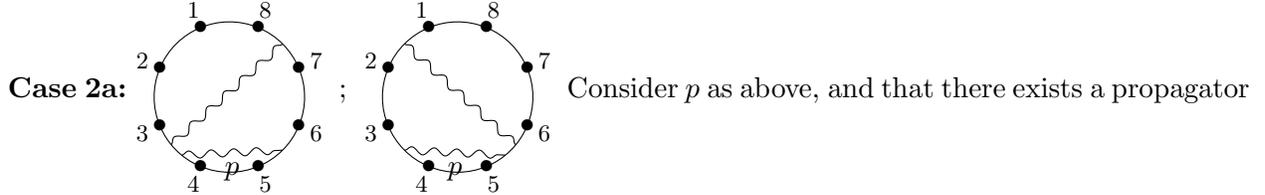
\textbf{Case 2a:}$  \begin{tikzpicture}[rotate=67.5,baseline=(current bounding box.east)] \begin{scope}
	\drawWLD{8}{1}
	\drawnumbers
	\drawlabeledprop{3}{1}{5}{0}{$p$}
        \drawlabeledprop{3}{-1}{7}{0}{}
        \end{scope} \end{tikzpicture} $ ; $\begin{tikzpicture}[rotate=67.5,baseline=(current bounding box.east)] \begin{scope}
	\drawWLD{8}{1}
	\drawnumbers
	\drawlabeledprop{3}{0}{5}{-1}{$p$}
        \drawlabeledprop{1}{0}{5}{1}{}
        \end{scope} \end{tikzpicture} $ Consider $p$ as above, and that there exists a propagator $(i+2, k)$ (resp. $(i, k)$) in $W$. If $x_{p, i}$ (resp. $x_{p, i+3}$) is a factor of $R(\VP)$, the singularity formed by sending $x_{p, i}$ (resp. $x_{p, i+3}$) to zero cancels with a pole coming from degree 2 factors contributed by other diagrams. Therefore, we return to this during the discussion of two by two minors. 

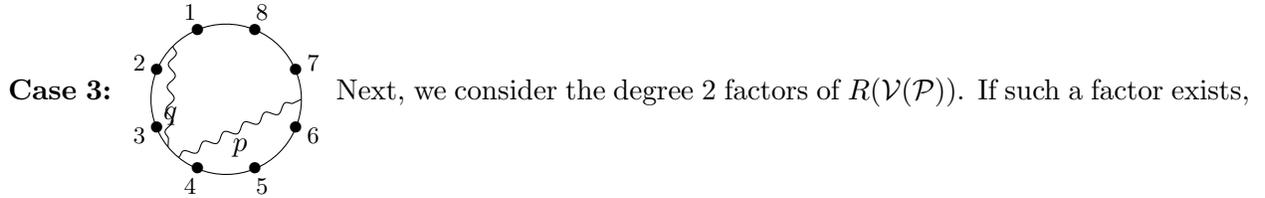
\textbf{Case 3:}  $\begin{tikzpicture}[rotate=67.5,baseline=(current bounding box.east)] \begin{scope}
	\drawWLD{8}{1}
	\drawnumbers
	\drawlabeledprop{3}{1}{6}{0}{$p$}
        \drawlabeledprop{9}{0}{3}{-1}{$q$}
        \end{scope} \end{tikzpicture} $ Next, we consider the degree 2 factors of $R(\VP)$. If such a factor exists, there must be two propagators $p = (i, j)$ and $q = (i, k)$ adjacent on the edge $i$. Suppose $k > j+1$, that is, the other endpoints of $p$ and $q$ are not on adjacent edges. If $r = (j,k) \not \in \cP$, consider two other diagrams formed by replacing the propagator $p$ and $q$ by r: $W' = (\cP' = (\cP \setminus p) \cup r, [n])$ and $W'' = (\cP'' = (\cP \setminus q) \cup r, [n])$. Since $k> j+1$ and $r \not \in \cP$, both $W'$ and $W''$ satisfy density. Furthermore, since $W$ is admissible, and $p$ and $q$ are adjacent on the edge $i$, there does not exist a propagator $(i, m)$ with $j < m <k$, that is, that has one endpoint on the $i^{th}$ edge, and the other between the other endpoints of $p$ and $q$. Therefore, $W'$ and $W''$ satisfy the non-crossing condition. Thus, $W'$ and $W''$ are both admissible. Note that the diagrams $W$, $W'$ and $W''$ are in the configuration laid out in display \eqref{eq:wideV}. By Corrollary \ref{res:reparamdim}, we see that the reparameterization performed in Lemma \ref{lem:simplifyR(W)} means that $\lim_{(x_{p,i}x_{q,i+1} -x_{p,i+1}x_{q,i})\rightarrow 0} M_{\VP}$  parameterizes a codimension one subspace of $\overline{\Sigma(\VP)}$. By Proposition \ref{res:vanishonbdny}, this is lies in a codimension one boundary of $\Sigma(\VP)$. According to Lemma \ref{res:Vdiagcancel}, after appropriate changes of parameterizations, one may write \bas \lim_{(x_{r,k}x_{q,k+1} -x_{r,k+1}x_{q,k})\rightarrow 0}I(\VP') +  \lim_{(x_{p,i}x_{q,i+1} -x_{p,i+1}x_{q,i})\rightarrow 0} I(\VP) + \lim_{(x_{p,j}x_{r,j+1} -x_{p,j+1}x_{r,j})\rightarrow 0}I(\VP'') = 0 \; .\eas 
In other words, these three singularities, which lie on the common codimension one boundary positroid of $\Sigma(\VP)$, $\Sigma(\VP')$ and $\Sigma(\VP'')$, cancel in the sum of integrals in the tree level amplitude.

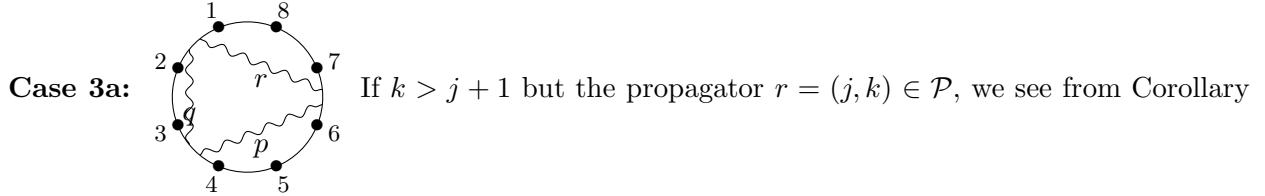
\textbf{Case 3a:} $\begin{tikzpicture}[rotate=67.5,baseline=(current bounding box.east)] \begin{scope}
	\drawWLD{8}{1}
	\drawnumbers
	\drawlabeledprop{3}{1}{6}{-1}{$p$}
        \drawlabeledprop{3}{-1}{1}{1}{$q$}
        \drawlabeledprop{6}{1}{1}{-1}{$r$}
        \end{scope} \end{tikzpicture} $ If $k > j+1$ but the propagator $r =  (j,k) \in \cP$, we see from Corollary \ref{res:reparamdim} that this has codimension greater than 1.  

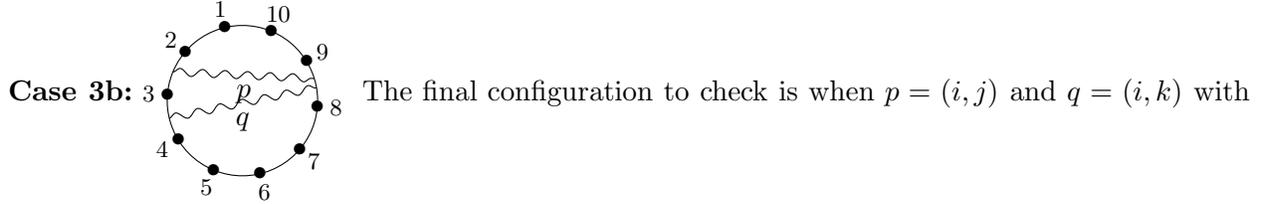
\textbf{Case 3b:}$\begin{tikzpicture}[rotate=67.5,baseline=(current bounding box.east)]
	\begin{scope}
	\drawWLD{10}{1}
	\drawnumbers
	\drawlabeledprop{3}{0}{8}{-1}{$q$}
	\drawlabeledprop{2}{0}{8}{1}{$p$}
		\end{scope}
	\end{tikzpicture} $ The final configuration to check is when $p = (i, j)$ and $q = (i, k)$ with $k = j+1$. Consider the diagrams $W' = (\cP' = (\cP \setminus p) \cup r = (j, j+2), [n])$ and $W'' = (\cP'' = (\cP \setminus p) \cup s = (k-2, k), [n])$. Since the edge $r \not \in \cP$ (it would cross $q$ if it were), and $s \not \in \cP$ (it would cross $p$ if it were), we see that $W'$ and $W''$ satisfy both the non-crossing and density conditions, and thus are admissible. Furthermore, $W$, $W'$ and $W''$ are in the configurations laid out in display \eqref{eq:narrowV} and the diagrams $W'$ and $W''$ are in the configuration laid out in Case 2a.

We see from Lemma \ref{res:narrowVcancel} that pole defined by the limit of sending $(x_{p,i}x_{q,i+1} -x_{p,i+1}x_{q,i})$ to $0$ cancels with degree one poles in the diagrams in $W'$ and $W''$ under the correct parameterization: \bas \lim_{(x_{p,i}x_{q,i+1} -x_{p,i+1}x_{q,i}) \rightarrow 0}I(\VP) + \lim_{x_{r, k-2}\rightarrow 0}I(\cV(\cP')) + \lim_{x_{s, j+3}\rightarrow 0}I(\cV(\cP'')) =0 \;.\eas The limits of $I(\cV(\cP'))$ and $I(\cV(\cP'))$ are both codimension one. By Corollary \ref{res:reparamdim}, we see that the limit of $I(\VP)$ is as well. 
\end{proof}

\begin{rmk}
Since this cancellation was demonstrated at point by point level, it holds not only on $\Gr(k,n)$, but also on $\Grall(k,n)$.
However, we remark that this cancellation is exact only in the space parameterized by matrices of the form $M(\VP)$, as a subspace of $\Grall(k,n)$ and not in the space parameterized by matrices of the form $M(\YP)$, as a subspace of $\Grall(k,n+1)$. In \cite{HeslopStewart,non-orientable}, the authors explicitly show that cancellations of this form do not hold in the larger space because of orientation issues.
\end{rmk}

\begin{appendices} 
\section{Pole Cancellation calculations \label{sec:appendix}} 
In this section, we present some useful calculations for understanding of the cancellation of spurious poles. Many of the results here can be found in \cite{casestudy, Amplituhedronsquared, HeslopStewart}. However, they are presented here for completeness.

Recall from Definition \ref{dfn:I(W)} that \bas \cI(\VP) (\cZ_*)  = \int_{(\RP^4)^k} \frac{\prod_{p \in \cP} \prod_{v \in V_p} dx_{p, v}}{R(\VP)} \delta^{4k|4k}(M_{\YP} \cdot \cZ_*) \eas where, for $X$ a $k \times k+4$ matrix, \bas \delta^{4k|4k}(X) = \prod_{b =1}^k (X_{b, 4+b})^4\delta^4((X_{b,1},X_{b,2},X_{b,3},X_{b,4}))  \;.\eas Write $\cZ_*^i$ to indicate the $i^{th}$ column of $\cZ_*$ and $\cZ_*^\mu$ to indicate the matrix formed by taking the first 4 columns of $\cZ_*$. Then evaluating the integral $I(\VP)$ corresponds to localizing the expression \bas \frac{\prod_{b = 1}^k (Y_b \cdot \cZ_*^b)^4}{R(\VP)}\eas at the solution to $M_{\YP} \cdot \cZ_*^\mu = 0$. Writing the propagator $p = (i, j)$ with $i <j$,  Cramer's rule implies that this localization evaluates to \ba x_{p, 0} &= \det(Z_i^\mu, Z_{i+1}^\mu, Z_{j}^\mu, Z_{j+1}^\mu ) \label{eq:matrixvalues1} \\ x_{p, i} = \det(Z_0^\mu, Z_{i+1}^\mu, Z_{j}^\mu, Z_{j+1}^\mu ) \; &; \; x_{p, i+1} = \det( Z_{i}^\mu, Z_0^\mu, Z_{j}^\mu, Z_{j+1}^\mu ) \; \text{ etc.} \label{eq:matrixvalues2}\ea That is, the entry $x_{p, m}$ evaluates to the minor of $\cZ_*^\mu$ indicated by the rows in $V_p$, with the $m^{th}$ row replaced by $Z_0^\mu$.

\begin{lem} \label{lem:movingpropnegative}
For two propagators $p = (i, j)$ and $q = (i, j+1)$, after localization $x_{p, j} = -x_{q, j+2}$.
\end{lem} 

\begin{proof}
By the above arguments, note that $x_{p, j} = \det(Z_i^\mu, Z_{i+1}^\mu, Z_{0}^\mu, Z_{j+1}^\mu )$ while $x_{q, j+2} = \det(Z_i^\mu, Z_{i+1}^\mu, Z_{j+1}^\mu , Z_{0}^\mu )$. Thus these two values are negatives of each other.
\end{proof}

Sometimes, as in Theorem \ref{res:polesonboundaries} and Theorem \ref{res:deg1polescancel}, it is necessary to perform changes of variables in order to understand the relevant loci. In particular, we need the following simplifying change of variables:
\begin{lem}\label{lem:simplifyR(W)}
Consider two propagators $p = (i, j)$ and $q = (i, k)$ that are adjacent on the $i^{th}$ edge of a Wilson loop diagram $(\cP, [n])$, with $p$ appearing closer to the vertex $i$ and $q$ closer to the vertex $i+1$. There is a reparameterization of the matrix $M_{\VP}$ under which one can replace the factor $x_{p, i}(x_{p, i}x_{q, i+1} - x_{p, i+1}x_{q, i})x_{q, i+1}$ in $R(\VP)$ with the product of 4 terms: $xyzw$.
\end{lem}

\begin{proof}
We restrict our attention to the relevant two by two minor of $M_{\VP}$, $ \begin{bmatrix} x_{p, i} & x_{p, i+1} \\ x_{q, i} & x_{q, i+1} \end{bmatrix} $,  which we can reparameterize as $ \begin{bmatrix} x & y \\ xz & zy + w \end{bmatrix} $. Then we have that \bas x_{p, i} = x \quad ; \quad  x_{p, i+1} = y \quad ; \quad x_{q, i} = xz\quad  ;\quad x_{q, i+1}  = zy + w \; . \eas Furthermore, \bas dx_{p, i} = dx  \quad ; \quad  d x_{p, i+1} = dy \quad ; \quad dx_{q, i} = x dz + z dx \quad ; \quad d x_{q, i+1} = ydz + z dy + dw \;.\eas Therefore, under these changes of variables, we see that \bas \frac{dx_{p, i}\;dx_{p, i+1}\;dx_{q, i}\;dx_{q, i+1}}{x_{p, i+1}(x_{p, i}x_{q, i+1} - x_{q, i}x_{p, i+1} ) x_{q, i}}  = \frac{ dx\;dy\;x dz\; dw}{y (xyz + xw - xyz)xz}\eas which simplifies to the desired result.
\end{proof}

In general, we use whichever parameterization of the two by two minors is convenient. To understand the cancelation of spurios poles, the need for a change of variables comes up in two cases. The first case involves the cancellation of the two by two minors in following three propagator configurations (see Case 3 for Theorem \ref{res:deg1polescancel}): \bml  \textrm{Config 1} = 
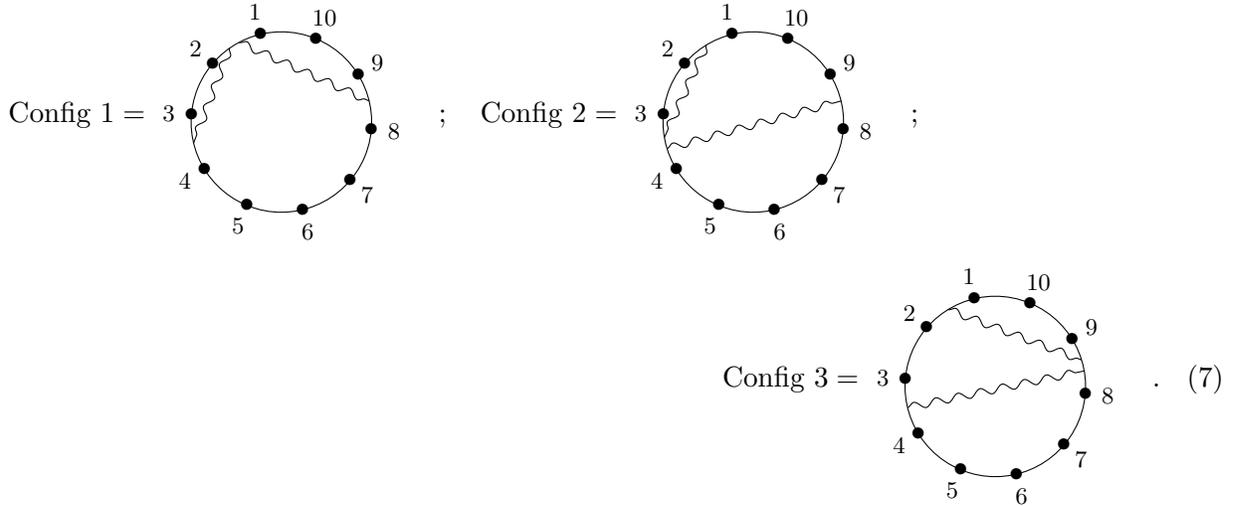
\begin{tikzpicture}[rotate=67.5,baseline=(current bounding box.east)]
	\begin{scope}
	\drawWLD{10}{1.2}
	\drawnumbers
	\drawprop{1}{-1}{8}{0}
	\drawprop{1}{1}{3}{0}
		\end{scope}
	\end{tikzpicture} \quad ; \quad \textrm{Config 2} = \begin{tikzpicture}[rotate=67.5,baseline=(current bounding box.east)]
	\begin{scope}
	\drawWLD{10}{1.2}
	\drawnumbers
	\drawprop{1}{0}{3}{-1}
	\drawprop{3}{1}{8}{0}
		\end{scope}
	\end{tikzpicture} \quad ; \\ \textrm{Config 3} = \begin{tikzpicture}[rotate=67.5,baseline=(current bounding box.east)]
	 \begin{scope}
	\drawWLD{10}{1.2}
	\drawnumbers
	\drawprop{1}{0}{8}{1}
	\drawprop{3}{0}{8}{-1}
		\end{scope}
	\end{tikzpicture}  \quad . \label{eq:wideV}\eml

\begin{lem}\label{res:Vdiagcancel}
Let $(\cP_1, [n])$, $(\cP_2, [n])$ and $(\cP_3, [n])$ be three different admissible Wilson loop diagrams that are identical except for the fact that the propagator set $\cP_i$ contains the pair of adjacent propagators shown in $\textrm{Config i}$ above. Then 
\bas \sum_{i = 1}^3 \lim_{\textrm{degree 2 factor of } R(\cV(\cP_i)) \rightarrow 0} I(\cV(\cP_i)) = 0\;.\eas \end{lem}

This proof is also given in \cite{HeslopStewart, casestudy} in a slightly different parametrization, and is included here for completeness.
\begin{proof}
Without loss of generality, write $M_{\cY(\cP_i)}$ with the pertinent propagators represented in the first two rows. Then the matrices $M_{\cY(\cP_i)}$ are identical except for the first two rows. Since the propagators are adjacent, by Lemma \ref{lem:simplifyR(W)} we may write the first two rows as 
\bas M_{\cY(\cP_1)} = \begin{bmatrix}1 & \ldots & a &b &\ldots & 0 & 0 & \ldots & c & d   \ldots\\  1 & \dots &  ae &be + f  & \ldots &g & h & \ldots &0 &0  \ldots   \end{bmatrix}  \\ 
M_{\cY(\cP_2)} = \begin{bmatrix}1 & \ldots & a' & b' &\ldots & c' & d' & \ldots & 0 & 0   \ldots\\  1 & \dots &  0 & 0  & \ldots &c'e' & d' e' + f' & \ldots &g' &h'  \ldots   \end{bmatrix} \\ 
M_{\cY(\cP_3)} = \begin{bmatrix}1 & \ldots & 0 &0 &\ldots & a'' & b'' & \ldots & c'' & d''   \ldots\\  1 & \dots &  e'' & f''  & \ldots &0 & 0 & \ldots & c''g'' & d''g'' + h''  \ldots   \end{bmatrix}\;.\eas 
We multiply the relevant rows of $M_{\cY(\cP_2)}$ and $M_{\cY(\cP_3)}$ by elements of $GL(2)$, leaving the rest of the rows unchanged. Namely, consider the products: 
\bas \begin{bmatrix} \frac{- e'}{1-e'} & \frac{1}{1-e'} \\ 1 & 0  \end{bmatrix} M_{\cY(\cP_2)} = \begin{bmatrix}  1 & \dots & \frac{-e' a'}{1- e'}  & \frac{-e' b'}{1- e'}  & \ldots &0 & \frac{ f'}{1-e'} & \ldots &\frac{g'}{1-e'} &\frac{h'}{1-e'}  \ldots  \\ 1 & \ldots & a' & b' &\ldots & c' & d' & \ldots & 0 & 0   \ldots \end{bmatrix}  
 \\ 
\begin{bmatrix} 0 & 1 \\   \frac{-g''}{1-g''} & \frac{1}{1-g''}\end{bmatrix} M_{\cY(\cP_3)} = \begin{bmatrix}1 & \ldots & e'' & f'' &\ldots & 0 & 0 & \ldots & c'' g'' & d''g'' + h''g''   \ldots\\  1 & \dots &  \frac{e''}{1-g''} & \frac{f''}{1-g''}  & \ldots &  \frac{-a''g''}{1-g''} & \frac{-b''g''}{1-g''}  & \ldots &0 & \frac{h''}{1-g''} \ldots   \end{bmatrix}\;.\eas

From this, we see that, in the limit $f \rightarrow 0$ and $M_{\cY(\cP_1)}$ and $f' \rightarrow 0$ for $\begin{bmatrix} \frac{- e'}{1-e'} & \frac{1}{1-e'} \\ 1 & 0  \end{bmatrix} M_{\cY(\cP_2)}$, we have the change of variables \bmls a = \frac{-e'a'}{1-e'} \quad ; \quad b = \frac{-e'b'}{1-e'} \quad ; \quad c = \frac{g'}{1-e'} \quad ; \quad d = \frac{d}{1-e'}  \quad ; \\ \quad e = \frac{e'-1}{e'}\quad ; \quad f = 0 \quad ; \quad g = c' \quad ; \quad h = d'\;.\emls Inverting and performing the change of variables, we see that $\lim_{f' \rightarrow 0} I(\cV(\cP_2)) = \lim_{f \rightarrow 0} \frac{e}{1-e} I(\cV(\cP_1))$. A similar calculation shows that $\lim_{h'' \rightarrow 0} I(\cV(\cP_3)) = \lim_{f \rightarrow 0} \frac{-1}{1-e} I(\cV(\cP_1))$. Thus, in the appropriate limit, \bas \lim_{f \rightarrow 0} I(\cV(\cP_1)) + \lim_{f' \rightarrow 0} I(\cV(\cP_2)) + \lim_{h'' \rightarrow 0} I(\cV(\cP_3)) = 0\;.\qedhere\eas  \end{proof} 

The last case to consider consists of understanding the poles shared between the diagrams with the following configurations: \bml \textrm{Config 4} = 
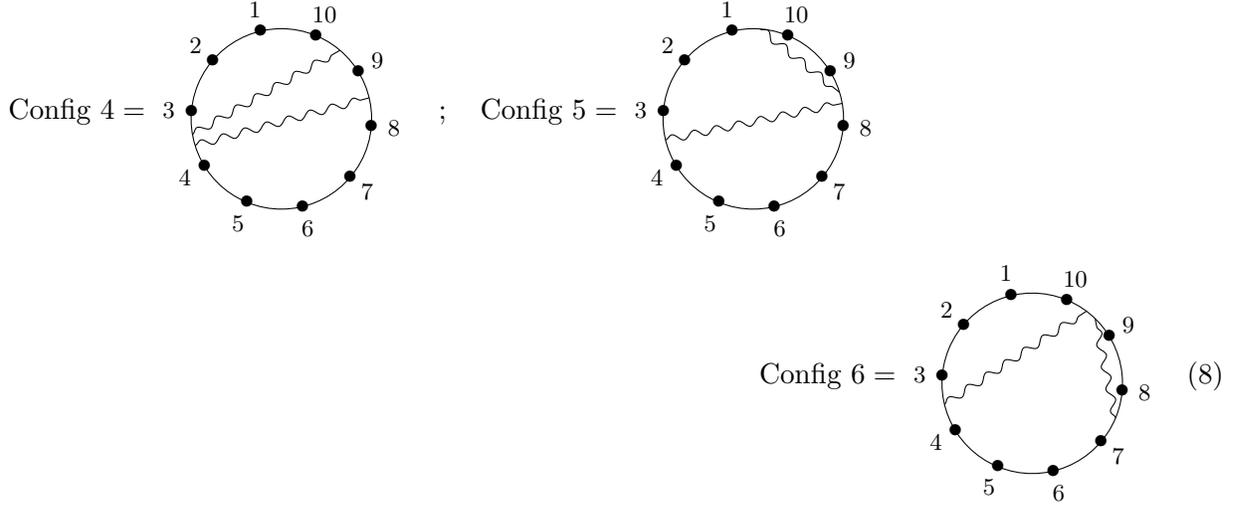
\begin{tikzpicture}[rotate=67.5,baseline=(current bounding box.east)]
	 \begin{scope}
	\drawWLD{10}{1.2}
	\drawnumbers
	\drawprop{3}{-1}{9}{0}
	\drawprop{3}{1}{8}{0}
		\end{scope}
	\end{tikzpicture} \quad ; \quad \textrm{Config 5} = \begin{tikzpicture}[rotate=67.5,baseline=(current bounding box.east)]
	\begin{scope}
	\drawWLD{10}{1.2}
	\drawnumbers
	\drawprop{3}{0}{8}{-1}
	\drawprop{10}{0}{8}{1}
		\end{scope}
	\end{tikzpicture} \\ \textrm{Config 6} = \begin{tikzpicture}[rotate=67.5,baseline=(current bounding box.east)]
	\begin{scope}
	\drawWLD{10}{1.2}
	\drawnumbers
	\drawprop{3}{0}{9}{1}
	\drawprop{7}{0}{9}{-1}
		\end{scope}
	\end{tikzpicture} \label{eq:narrowV}\eml

\begin{lem}\label{res:narrowVcancel}
Let $(\cP_4, [n])$, $(\cP_5, [n])$ and $(\cP_6, [n])$ be admissible Wilson loop diagrams that are identical except for the fact that the propagator set $\cP_i$ contains the pair of adjacent propagators shown in $\textrm{Config i}$ above. Let $p = (i, j)$ and $q = (i, k)$ with $k = j+1$. Then \bas \lim_{(x_{p,i}x_{q, i+1} - x_{p, i+1}, x_{q, i}) \rightarrow 0} I(\cV(\cP_4)) + \lim_{x_{r, j+3} \rightarrow 0}I(\cV(\cP_5)) + \lim_{x_{r, k-2} \rightarrow 0}I(\cV(\cP_6)) = 0\;.\eas \end{lem}
\begin{proof}
This proof follows similarly to the above. Write \bas M_{\cY(\cP_4)} = \begin{bmatrix}1 & \ldots & a &b &\ldots & c & d & 0 & \ldots \\  1 & \dots &  ae &be + f  & \ldots &0 & g & h & \ldots    \end{bmatrix}  \\ M_{\cY(\cP_5)} = \begin{bmatrix}1 & \ldots & a' &b' &\ldots & c' & d' & 0 & 0& \ldots \\  1 & \dots &  0&0  & \ldots &c'e' & d'e' +f'  & g' & h' & \ldots    \end{bmatrix} \\ M_{\cY(\cP_6)} = \begin{bmatrix}1 & \ldots  & 0 &0 &\ldots  & a'' & b'' & c'' &  d'' & \ldots \\  1 & \dots & e'' & f''  & \ldots  &0&0 & c'' g'' & d'' g'' + h'' & \ldots    \end{bmatrix} \;. \eas We consider the change of variables defined by taking the product with \bas \begin{bmatrix}1 & 0 \\ \frac{-e'}{1-e'} & \frac{1}{1-e'} \end{bmatrix} M_{\cY(\cP_5)}  \quad \textrm{ and } \quad \begin{bmatrix}0 & 1 \\ \frac{-g''}{1-g''} & \frac{1}{1-g''}  \end{bmatrix} M_{\cY(\cP_6)} \;.\eas Then the same types of calculations as in Lemma \ref{res:Vdiagcancel} shows that $ \lim_{a'' \rightarrow 0}I(\cV(\cP_6))  =  \lim_{f \rightarrow 0} \frac{-1}{1-e} I(M_{\cV(\cP_4)}) $ and $\lim_{h' \rightarrow 0}I(\cV(\cP_5))  =  \lim_{f \rightarrow 0} \frac{e}{1-e} I(\cV(\cP_4)) $, proving the result.
\end{proof}


Finally, we show that the limits defined in Lemma \ref{res:Vdiagcancel} and Lemma \ref{res:narrowVcancel} do in fact give rise to codimension one boundaries of $\Sigma(\VP)$. To do this, we define a more general result. In Theorem \ref{res:Rado}, we show that, if $M'$ is a variable valued matrix formed by applying an invertible change of variables to a matrix $M_\cV$, then sending $k$ variables to $0$ in $M'$ drops the dimension of the parameterized space $L(M')$ if and only if no row of $M'$ is contained in the span of any other subset of rows of $M'$.

The result is a generalization of the equivalence of 1 and 3 in Theorem \ref{res:minimalrep}. Condition 3 from Theorem \ref{res:minimalrep} is analogous to the inequality from Hall's Matching Theorem. However, 
the matrices in Theorem \ref{res:minimalrep} have the restriction that entries are either zero or totally independent. Theorem \ref{res:Rado} below, relaxes this condition, allowing one to view each independent non-zero entry as a function of several variables, and setting one of these component variables to $0$, as in Lemma \ref{res:Vdiagcancel} and Lemma \ref{res:narrowVcancel}. 

We begin with a few definitions. 

\begin{dfn}
For a $k \times n$ variable valued matrix $M$ and a set $I \subseteq [k]$, define $M_I$ to be the matrix formed by restricting $M$ to the rows $I$, and let $\mathrm{span}(M_I)$ be the subset of $\mathbb{R}^n$ obtained by evaluating linear combinations of the rows $M$ indexed by $I$ at real parameters. Let $\mathrm{span}(M_{\emptyset})$ be the origin in $\mathbb{R}^n$.
\end{dfn}

Note that these matrices are different from the variable valued matrices defined by $M_{\cV}$ above. We relax the requirement that each entry of the matrix be algebraically independent of all the others. Furthermore, while the entries of $M_I$ are functions of independent variables, the entries themselves may be related by algebraic functions.

\begin{dfn}
Let $L(M_I)$ be the locus of the variable valued matrix $M_I$ in $\Grall(|I|, n)$, as given by the map in \eqref{eq:maps}.
\end{dfn}

\begin{eg}
Let
\begin{displaymath}
M =
\begin{bmatrix}x_1 & x_2 & 0 & 0 \\
0 & 0 & x_1 & 0 \\
0 & 0 & x_2 & x_1 \end{bmatrix}.
\end{displaymath}
\noindent
Then, $\mathrm{span}(M_{12})$ consists of all points in $\mathbb{R}^4$ whose last coordinate is zero. On the other hand, $\mathrm{span}(M_{13})$ consists of points of the form $(a,b,\frac{bc}{a},c)$ for some $a \in \mathbb{R} \setminus \{0\}$ and $b, c \in \mathbb{R}$ together with points of the form $(0,b,c,0)$. Finally, $\mathrm{span}(M_{123})$ is all of $\mathbb{R}^4$.
\end{eg}

\begin{thm}\label{res:Rado} 
Let $\cV$ be a collection of subsets, and let $\mathbf{x} = \{x_{i,j}\}$ and $\mathbf{y} = \{y_{i,j}(\mathbf{x})\}$ be two sets of algebraically independent invertible variables (arising from the same indexing set $\cV$) which are related by a change of variables. Denote by $M_{\cV}(\mathbf{x})$ and $M_{\cV}(\mathbf{y})$ the variable valued matrices associated to $\cV$ in the variables $\mathbf{x}$ and $\mathbf{y}$. Let $S \subset \mathbf{x}$ be a subset of the variables $\bf{x}$, and write $M' = M_{\mathcal{P}}(\mathbf{y})|_{x_{i,j} = 0; x_{i,j} \in S}$. If the function $\lim_{x_{i,j} \in S \to 0} \bf{y}(\bf{x})$ is invertible, denote by $d = |\mathbf{x}| - |S|$ the number of variables in $\mathbf{x}$ outside of the set $S$. 

The following are equivalent: 
\begin{itemize}
\item[(i)] For all $I \subsetneq [k]$ and all $j \in I^c$, $\mathrm{span}(M'_I) \subsetneq \mathrm{span}(M'_{I \cup j})$. That is, adding a row to $M'_I$ always increases the size of the span.
\item[(ii)] $\mathrm{dim}(L(M')) = d-k$.
\end{itemize}
\end{thm}

\begin{rmk}
Similarly to how Condition 3 in Theorem \ref{res:minimalrep} should be thought of as analogous to the inequality from Hall's Matching Theorem, condition (i) above should be thought of as analogous to the condition from Rado's Theorem. Rado's Theorem is a generalization of Hall's Theorem, which says that if $S_1, \dots, S_k \subseteq \mathbb{R}^{n}$, then it is possible to select linearly independent vectors $s_1 \in S_1, \dots, s_k \in S_k$ if and only if for all $I \subseteq [k]$,
\begin{displaymath}
\mathrm{dim}\left(\mathrm{span}\left( \bigcup_{i \in I} S_i \right) \right) \geq |I|.
\end{displaymath}
\noindent
Condition (ii) from Theorem \ref{res:Rado} relaxes this condition from Rado's theorem, saying not only does any subset of vectors have the correct dimension, but adding a new vector always increases the dimension.
\end{rmk}

\begin{proof}
We show that (ii) implies (i) via induction on $k$, the number of rows of $M_\cV(\bf{x})$. When $k = 1$, $L(M')$ is a parameterized subset of projective space. Since the transformation from the $\mathbf{x}$ to the $\mathbf{y}$ variables are invertible on $M' = M_{\cP}(\mathbf{y})|_{x_{i,j} = 0; x_{i,j} \in S}$, the matrices $M'$ and $M_{\cP}(\mathbf{x})|_{x_{i,j} = 0; x_{i,j} \in S}$ have the same dimension in $\R^n$. One obtains the corresponding locus in projective space by scaling by a constant, so the dimensions in projective space remain the same. 

For larger $k$, it is not sufficient to use the invertibility of the change of variables as an argument that the dimension does not change. While this holds on the level of matrices, it may not hold after quotienting by $GL(k)$ in order to find the locus $L(M')$.

In this proof, we consider the Grassmannian set theoretically. So $\Grall(k, \mathbb{R}^n)$ is the set of $k$-planes in $\mathbb{R}^n$. Given a vector space $V \subsetneq \R^n$,  $(\mathbb{R}^n \perp V)$ is a subset of $\mathbb{R}^n$, and we write 
\begin{displaymath}
\Grall(k, \mathbb{R}^n \perp V) = \{x \in  \Grall(\mathbb{R}^n,k)| x \in \R^n \perp V\}.
\end{displaymath}
\noindent

Suppose that the result holds for $k = l-1$. Let $M'$ be an $l \times n$ matrix such that $\mathrm{span}(M'_{\{1,\dots,l\}}) = \mathrm{span}(M'_{1,\dots,l-1})$, but that $\mathrm{span}(M'_I) \neq \mathrm{span}(M'_{I \cup j})$ for $j \notin I$ whenever $|I| < l$. That is, condition (i) fails to hold only for the sets $|I| = l-1$. 

A plane in $L(M')$ is spanned by a nonzero vector $r \in  \mathrm{span}(M'_{l})$, plus a $l-1$ plane in
\begin{displaymath}
L(M'_{1,\dots,l-1}) \cap \Grall(k-1, \mathbb{R}^n \perp r),
\end{displaymath}
i.e.  an $(l-1)$-plane in $L(M')$ which is orthogonal to $r$.
Since condition (i) does not hold for $I = \{1,\dots,l-1\}$, $r \in \mathrm{span}(M'_{1, \dots, l-1})$. As the map \eqref{eq:maps} ignores the points parameterized by $M'_I$ that do not have full rank, we do not need to consider the case when $\mathrm{span}(M'_{1,\dots,l-1}) = \mathrm{span}(r)$ Observe that  the set of planes in $L(M'_{1,\dots,l-1})$ containing $r$ has positive dimension.
Thus,
\begin{displaymath}
\dim(L(M'_{1,\dots,l-1}) \cap \Grall(k-1, \mathbb{R}^n \perp r)) < \dim(L(M'_{1,\dots,l-1})). 
\end{displaymath}
\noindent
Note that the subspace of $\Grall(1,\mathbb{R}^n) \times \Grall(k-1, \mathbb{R}^n)$ obtained by evaluating $M'$ and independently taking the span of the $l$th row and of rows $1, \dots, l-1$ (ignoring matrices of improper rank) is $d-k$ by the inductive hypothesis. By the remarks above, $L(M')$ has strictly lower dimension than this set. So,
\begin{displaymath}
\dim(L(M')) \leq  d -k - 1,
\end{displaymath}
\noindent
and thus (ii) implies (i) when $k = l$.

Next, we show that (i) implies (ii). When $k = 1$, $L(M')$ is again a parameterized subset of projective space, and the argument holds similarly to the case for $k=1$ in the other direction of this proof. 

Suppose the result holds for $k = l-1$ and let $M'$ be an $l \times n$ matrix such that $\mathrm{span}(M'_{1,\dots,l}) \neq \mathrm{span}(M'_{1,\dots,l-1})$. Let $d_l$ be the number of parameters appearing in row $l$ of $M'$ which do not appear in rows $1, \dots, l-1$, i.e. there are $d - d_l$ variables in the first $l-1$ rows of $M'$. Note that $L(M'_l)$ is $d_l -1$ dimensional by similar arguments as the base case.

Then, for a generic point $r \in \mathrm{span}(M'_{l})$ that comes from evaluating the last row of $M'$, note that $r \notin \mathrm{span}(M'_{1,\dots,l-1})$. So, $L(M'_{1,\dots,l-1})\cap \Grall(k-1, \mathbb{R}^n \perp r) = L(M'_{1,\dots,l-1})$. In other words
\begin{displaymath}
\dim(L(M'_{1,\dots,l-1}) \cap \Grall(k-1, \mathbb{R}^n \perp r)) = \dim(L(M'_{1,\dots,l-1})) = d - d_l -(l-1)\;,
\end{displaymath}
\noindent
where, by induction $\dim(L(M'_{1,\dots,l-1})) = d - d_l -(l-1)$. Thus,
\begin{displaymath}
\dim(L(M'_{1,\dots,l})) = (d_l - 1) + \dim(L(M'_{1,\dots,l-1})) = d-l.
\qedhere
\end{displaymath}
\end{proof}

Given Theorem \ref{res:Rado}, we see that the boundaries defined in Lemma \ref{res:Vdiagcancel} and Lemma \ref{res:narrowVcancel} do in fact give rise to codimension one boundaries of $\Sigma(\VP)$.

\begin{cor} \label{res:reparamdim}
Let $(\cP, [n])$ be a diagram with two propagators $p = (i, j)$ and $q = (i, k)$ adjacent on the edge $i$. After the reparameterization defined in Lemma \ref{lem:simplifyR(W)}, setting $w$ to zero reduces the dimension by one if and only if there is not a propagator $r = (j, k)$ in $\cP$.
\end{cor}
\begin{proof}
For a Wilson loop diagram $(\cP, [n])$, let $\bf{x}$ be the parameterization that includes the variable $w$, and let $\bf{y}$ be the other set of variables. Note that the Jacobian for the change of variables given in Lemma \ref{lem:simplifyR(W)} has zero determinant when $x$ is sent to $0$. When $w$ is set to zero, the determinant of the Jacobian remains non-zero. Without loss of generality, assume that $w$ is a variable in the row corresponding to the propagator $p$.

Further note that whether or not $r$ is in $\cP$, the diagram $(\cP \cup r, [n])$ is admissible. If $r$ does not exist in $\cP$, then adding it does not violate density. Furthermore, because $p$ and $q$ are adjacent, adding $r$ does not violate the non-crossing condition either. 

First we show that when $r \in \cP$, when one sets $w = 0$ the dimension of the span of a row set does not increase with the rows included in the set. In particular, we may write $x_{q, i} = \lambda x_{p, i}$ and $x_{q, i+1} = \lambda x_{p, i+1}$, with $\lambda$ a real variable. By adding a scalar multiple of $M'_p$ to $M'_q$, one gets a row that has independent variable entries exactly in the columns $V_r =\{j, j+1, k, k+1\}$. If $r$ is in $\cP$, the span of the rows $p$ and $q$ contains the span of the row $r$ in $\R^n$:
\ba \mathrm{span}(M'_{\{p,q\}}) =  \mathrm{span}(M'_{\{p,q,r\}})\;.\label{eq:spanfact}\ea Thus, Theorem \ref{res:Rado} implies this boundary does not have codimension one.

Next, suppose $r \not \in \cP$. In this case, assume toward contradiction that there is a set of propagators $P \subset \cP$ such that $p \not \in P$ and  $\mathrm{span}(M'_{P \cup p}) = \mathrm{span}(M'_P)$. Without loss of generality, assume that $P$ one of the smallest sets of propagators such that this is true. Note that $P$ cannot exist if $q \not \in P$, since all the variables in $M'_p$ and $M'_q$ are algebraically independent from the variables in $M'_{P\setminus \{p,q\}}$. Therefore, without loss of generality, assume $q \in P$. 

Since $\mathrm{span}(M'_{p}) \neq \mathrm{span}(M'_{q})$, the equality $\mathrm{span}(M'_{P \cup p}) = \mathrm{span}(M'_P)$ implies that  $\mathrm{span}(M'_{\{p, q\}}) \subseteq \mathrm{span}(M'_P)$. But equation \eqref{eq:spanfact} means that  $\mathrm{span}(M'_{\{p,q,r\}}) \subseteq \mathrm{span}(M'_P)$. Removing the rows $p$ and $q$ from both sides gives
\ba \mathrm{span}(M'_{r})  \subseteq \mathrm{span}(M'_{P\setminus q}) \label{eq:spandfact2}\;.\ea

We claim that this cannot be true. Since $q \in P$ but $p$ is not, we have the equality $M'_{P\setminus q} = M_{\cV(P\setminus q)}$ because the propagator set $P \setminus q$ is unaffected by the limits defining the matrix $M'$. Note that the propagator set $(P\setminus q) \cup r$ is an admissible Wilson loop diagram. By Theorem \ref{res:minimalrep},
\bas \dim(L(\cV((P\setminus q) \cup r))) = \dim(L(\cV(P\setminus q))) + 3
,\eas
which contradicts \eqref{eq:spandfact2}. Therefore, condition (i) of Theorem \ref{res:Rado} must hold, and thus the limit $w \rightarrow 0$ defines a codimension one subspace of $L(M_{\VP})$.
\end{proof}

\end{appendices}

\bibliographystyle{unsrt}
\bibliography{Bibliography}

\end{document}